\documentclass[a4paper]{article} 

\newcommand{\longversion}[1]{} 
\newcommand{\shortversion}[1]{#1} 
 
\usepackage[lined,commentsnumbered,vlined,linesnumbered,boxed]{algorithm2e}
\usepackage{dsfont}
\usepackage{amsmath,amssymb,amsthm,amsfonts}
\usepackage{nicefrac}
\usepackage[margin=1in]{geometry}
\usepackage{authblk}

\newtheorem{theorem}{Theorem}[]
\newtheorem{lemma}{Lemma}[]

\usepackage{tikz}
\usetikzlibrary{automata, positioning}
\usetikzlibrary{shapes}
\tikzstyle{every node}=[minimum size=2mm]
\tikzstyle{vertex}=[circle,fill=black,inner sep=0mm,draw]

\usepackage{thmtools}
\usepackage{thm-restate}
\usepackage{refcount}
\newcounter{tempcounter}
\newcommand{\restateLemma}[2]{%
 \setcounter{tempcounter}{\value{theorem}}
 \setcounterref{theorem}{#2}
 \addtocounter{theorem}{-1}
 #1*
 \setcounter{theorem}{\value{tempcounter}}
}
\newtheorem{claim}{Claim}

\newcounter{tempcounterClaim}

\newcommand{\rteightthirds}{1.0963}
\newcommand{\rtcubic}{1.1389}
\newcommand{\rtquad}{1.2070}
\newcommand{\rtis}{1.2356}
\newcommand{\rtcol}{2.2356}
\newcommand{\bs}{\setminus}
\newcommand{\esrtis}{1.2330}

\title{Faster Graph Coloring in Polynomial Space}

\author[1,2]{Serge Gaspers}
\author[1,2]{Edward J. Lee}
\affil[1]{UNSW Australia, Sydney, Australia.\\ \texttt{sergeg@cse.unsw.edu.au}, \texttt{e.lee@unsw.edu.au}}
\affil[2]{Data61, CSIRO, Sydney, Australia}

\begin{document}

\maketitle

\begin{abstract}
We present a polynomial-space algorithm that computes the number of independent sets of any input graph in time $O(\rtcubic^n)$ for graphs with maximum degree 3 and in time $O(\rtis^n)$ for general graphs, where $n$ is the number of vertices\longversion{ in the input graph}.
Together with the inclusion-exclusion approach of Bj{\"o}rklund, Husfeldt, and Koivisto [SIAM J. Comput. 2009], 
this leads to a faster polynomial-space algorithm for the graph coloring problem with running time $O(\rtcol^n)$.
As a byproduct, we also obtain an exponential-space $O(\esrtis^n)$ time algorithm for counting independent sets.

Our main algorithm counts independent sets in graphs with maximum degree \longversion{at most }3 and no vertex with three neighbors of degree 3.
This polynomial-space algorithm is \longversion{designed and }analyzed using the recently introduced Separate, Measure and Conquer approach [Gaspers \& Sorkin, ICALP 2015].
Using Wahlstr{\"o}m's compound measure approach, this improvement in running time for small degree graphs is then bootstrapped to larger degrees,
giving the improvement for general graphs.
Combining both approaches leads to some inflexibility in choosing vertices to branch on for the small-degree cases, which we counter by structural graph properties.
The main complication is to upper bound the number of times the algorithm has to branch on vertices all of whose neighbors have degree 2, while still decreasing the size of the separator each time the algorithm branches.
\end{abstract}

\section{Introduction}

Graph coloring is a central problem in discrete mathematics and computer science.
In exponential time algorithmics \cite{FominK10}, graph coloring is among the most well studied problems, and it is an archetypical partitioning problem.
Given a graph $G$ and an integer $k$, the problem is to determine whether the vertex set of $G$ can be partitioned into $k$ independent sets.
Already in 1976, Lawler \cite{Lawler76} designed a dynamic programming algorithm for graph coloring and upper bounded its running time by $O(2.4423^n)$, where $n$ is the number of vertices of the input graph.
This was the best running time for graph coloring for 25 years, when Eppstein \cite{Eppstein01,Eppstein03} improved the running time to $O(2.4150^n)$ by using better bounds on the number of small maximal independent sets in a graph.
Based on bounds on the number of maximal induced bipartite subgraphs and refined bounds on the number of size-constrained maximal independent sets, Byskov \cite{Byskov04} improved the running time to $O(2.4023^n)$.
An algorithm based on fast matrix multiplication by Bj{\"o}rklund and Husfeldt \cite{BjorklundH08} improved the running time to $O(2.3236^n)$.
The current fastest algorithm for graph coloring, by Bj{\"o}rklund et al. \cite{BjorklundH06,BjorklundHK09,Koivisto06}, is based on the principle of inclusion--exclusion and Yates' algorithm for the fast zeta transform.
This breakthrough algorithm solves graph coloring in $O^*(2^n)$ time, where the $O^*$-notation is similar to the $O$-notation but ignores polynomial factors.


A significant drawback of the aforementioned algorithms is that they use exponential space. Often, the space bound is the same as the time bound, up to polynomial factors. This is undesirable \cite{Woeginger08}, certainly for modern computing devices.
Polynomial-space algorithms for graph coloring have been studied extensively as well with successive running times $O^*(n!)$ \cite{Christofides71}, $O((k/e)^n)$ (randomized) \cite{FederM02}, $O((2+ \log k)^n)$ \cite{AngelsmarkT05}, $O(5.283^n)$ \cite{BodlaenderK06}, $O(2.4423^n)$ \cite{BjorklundH08}, and $O(2.2461^n)$ \cite{BjorklundHK09}.
The latter algorithm is an inclusion--exclusion algorithm relying on a $O(1.2461^n)$ time algorithm \cite{FurerK07} for computing the number of independent sets in a graph\longversion{ as a subroutine}.
Their method transforms any polynomial-space $O(c^n)$ time algorithm for counting independent sets into a polynomial space $O((1+c)^n)$ time algorithm for graph coloring.
The running time bound for counting independent sets was subsequently improved by Fomin et al. \cite{fomin2009two} to $O(1.2431^n)$ and by Wahlstr{\"o}m \cite{Wahlstrom08} to $O(1.2377^n)$.
Wahlstr{\"o}m's algorithm is the current fastest published algorithm for counting independent sets of a graph, it uses polynomial space, and it works for the more general problem of computing the number of maximum-weight satisfying assignments of a 2-CNF formula.
For a reduction from counting independent sets to counting maximum-weight satisfying assignments of a 2-CNF formula where the number of variables equals the number of vertices, see \cite{DahllofJW05}.

We note that Junosza-Szaniawski and Tuczynski \cite{Junosza-SzaniawskiT15} present an algorithm for counting independent sets with running time $O(1.2369^n)$ in a technical report that also strives to disconnect low-degree graphs.
For graphs with maximum degree $3$ that have no degree-3 vertex with all neighbors of degree $3$, they present a new algorithm with running time $2^{n_3/5+o(n)}$, where $n_3$ is the number of degree-3 vertices, and the overall running time improvement comes from plugging this result into Wahlstr{\"{o}}m's \cite{Wahlstrom08} previously fastest algorithm for the problem.
However, we note that the $2^{n_3/5+o(n)}$ running time for counting independent sets can easily be obtained from previous results.
Namely, the problem of counting independent sets is a polynomial PCSP with domain size 2, as shown in \cite{scott2009polynomial}, and the algorithm of \cite{GaspersS15} for polynomial PCSPs preprocesses all degree-2 vertices, leaving a cubic graph on $n_3$ vertices that is solved in $2^{n/5+o(n)}$ time.
Improving on this bound is challenging, and degree-3 vertices with all neighbors of degree 2 need special attention since branching on them affects the degree-3 vertices of the graph exactly the same way as for the much more general polynomial PCSP problem, whereas for other degree-3 vertices one can take advantage of the asymmetric nature of the typical independent set branching (i.e., we can delete the neighbors when counting the independent sets containing the vertex we branch on).

\subparagraph{Our Results.}
We present a polynomial-space algorithm computing the number of independent sets of any input graph $G$ in time $O(\rtis^n)$, where $n$ is the number of vertices of $G$.
Our algorithm is a branching algorithm that works initially similarly as Wahlstr{\"o}m's algorithm, where we slightly improve the analysis using potentials (as, e.g., in \cite{GaspersS12,Iwata11,Wahlstrom04}) to amortize some of the worst branching cases with better ones.
This algorithm uses a branching strategy that basically ensures that both the maximum degree and the average degree of the graph do not increase.
This makes it possible to divide the analysis of the algorithm into sections depending on what local structures can still occur in the graph, use a separate measure for the analysis of each section, and combine these measures giving a compound (piecewise linear) measure for the analysis of the overall algorithm.

For instances where the maximum degree is 3 and no vertex has three neighbors with degree 3,
we substitute a subroutine that is designed and analyzed using the recently introduced \emph{Separate, Measure and Conquer} technique \cite{GaspersS15}.
\longversion{In this subroutine, the average degree of the graph is at most $\nicefrac{8}{3}$. }It computes a small balanced separator of the graph and prefers to branch on vertices in the separator, adjusting the separator as needed by the analysis, and reaping a huge benefit when the separator is exhausted and the resulting connected components can be handled independently.
The Separate, Measure and Conquer technique helps to amortize this sudden gain with the analysis of the previous branchings, for an overall improvement of the running time.

Since using a separator restricts our choice in the vertices to branch on, we use the structure of the graph and its separation to upper bound the number of unfavorable branching situations and adapt our measure accordingly.
Namely, the algorithm avoids branching on degree-3 vertices in the separator with all neighbors of degree 2 as long as possible, often rearranging the separator to avoid this case.
In our analysis we can then upper bound the number of unfavorable branchings and give the central vertex involved in such a branching a special weight and role in the analysis.
We call these vertices \emph{spider vertices}.
Our meticulous analysis of this subroutine upper bounds its running time by $O(\rteightthirds^n)$.
For graphs with maximum degree at most $3$, we obtain a running time of $O(\rtcubic^n)$. 
This improvement for small degree graphs is bootstrapped, using Wahlstr{\"o}m's compound measure analysis, to larger degrees, and gives a running time improvement to $O(\rtis^n)$ for counting independent sets of arbitrary graphs and to $O(\rtcol^n)$ for graph coloring.
Bootstrapping an exponential-space pathwidth-based $O(1.1225^n)$ time algorithm \cite{FominH06} for cubic graphs instead, we obtain an exponential-space algorithm for counting independent sets with running time $O(\esrtis^n)$.
\shortversion{ 
Some proofs have been moved to the appendix due to space constraints.
}

\section{Methods}

\subparagraph{Measure and Conquer.}
The analysis of our algorithm is based on the Measure and Conquer method \cite{FominGK09}.
A \emph{measure} for a problem (or its instances) is a function from the set of all instances of the problem to the set of non-negative reals.
Modern branching analyses often use a potential function as measure that gives a more fine-grained
way of tracking the progress of a branching algorithm than a measure that is merely the number of vertices or edges of the graph.
The following lemma is at the heart of our analysis.
It generalizes a similar lemma from \cite{GaspersS12} to the treatment of subroutines.

\begin{lemma}[\cite{Gaspers10}] \label{lem:combinemeasureanalysis}
    Let $A$ be an algorithm for a problem $P$,
    $B$ be an algorithm for a class $\mathcal{C}$ of instances of $P$,
    $c \ge 0$ and $r>1$ be constants,
    and $\mu(\cdot), \mu_B(\cdot), \eta(\cdot)$ be measures 
    for $P$,
    such that
    for any input instance $I$ from $\mathcal{C}$, $\mu_B(I) \le \mu(I)$, and
    for any input instance $I$,
    $A$ either solves $P$ on $I\in \mathcal{C}$ by invoking $B$ with running time $O(\eta(I)^{c+1} r^{\mu_B(I)})$,
    or reduces $I$ to $k$ instances $I_1,\ldots,I_k$,
    solves these recursively, and combines their solutions to solve~$I$,
    using time $O(\eta(I)^{c})$ for the reduction and combination steps
    (but not the recursive solves),
    \begin{align}
        (\forall i) \quad \eta(I_i) & \leq \eta(I)-1 \text{, and} \label{eq:eta}
        \\
        \sum_{i=1}^k r^{\mu(I_i)} & \leq r^{\mu(I)} . \label{eq:mu}
    \end{align}
    Then $A$ solves any instance $I$
    in time $O(\eta(I)^{c+1} r^{\mu(I)})$.
\end{lemma}

\noindent
When Algorithm $A$ does not invoke Algorithm $B$, we have the usual Measure and Conquer analysis.
Here, $\mu$ is used to upper bound the number of leaves of the search tree and deserves the most attention,
while $\eta$ is usually a polynomial measure \longversion{that is used }to upper bound the depth of the search tree.
For handling subroutines, it is crucial that the measure does not increase when Algorithm $A$ hands over
the instance to Algorithm $B$ and we constrain that $\mu_B(I)\le \mu(I)$.

\subparagraph{Compound analysis.}
We can view Wahlstr{\"o}m's compound analysis \cite{Wahlstrom08} as a repeated application of Lemma~\ref{lem:combinemeasureanalysis}. For example, there is one subroutine $A_3$ for when the maximum degree of the graph is 3. The algorithm prefers then to branch on a degree-3 vertex with all neighbors of degree 3. After all such vertices have been exhausted, the algorithm calls a new subroutine $A_{8/3}$ that takes as input a graph with maximum degree 3 where no degree-3 vertex has only degree 3 neighbors. In this case the average degree of the graph is at most $\nicefrac{8}{3}$, and the algorithm prefers to branch on vertices of degree 3 that have 2 neighbors of degree 3, etc. The analysis constrains that the measure for the analysis of $A_{8/3}$ is at most the measure for $A_{3}$ for the instance that is handed by $A_3$ to $A_{8/3}$.
In an optimal analysis, we expect the measure for such an instance to be equal in the analysis of $A_3$ and $A_{8/3}$, and Wahlstr{\"o}m actually imposes equality at the \emph{pivot point} $\nicefrac{8}{3}$.

\subparagraph{Separate, Measure and Conquer.}
In our case, the $A_{8/3}$ algorithm is based on \emph{Separate, Measure and Conquer}. 
For small-degree graphs, we can compute small balanced separators in polynomial time.
The algorithm then prefers to branch on vertices in the separator.
The Separate, Measure and Conquer technique allows to distribute the large gain obtained by disconnecting the instance onto the previous branching vectors.
While, often, the measure is made up of weights that are assigned to each vertex, this method assigns these weights only to the larger part of the graph that is separated from the rest by the separator, and somewhat larger weights to the vertices in the separator.
See \eqref{eq:deg3measure} on page \pageref{eq:deg3measure}.
Thus, after exhausting the separator, the measure accurately reflects the ``amount of work'' left to do.
We artificially increase the measure of very balanced instances by small penalty weights -- this is so because branching on vertices can change the measure of the parts that are separated by the separator and the branching strategy might not always be able to make most of its progress on the large side.
Since we may exhaust the separators a logarithmic number of times, and computing a new separator might introduce a penalty term each time, the measure also includes a logarithmic term that counteracts these artificial increases in measure, and will in the end only contribute a polynomial factor to the running time.
For an in-depth treatment of the method we refer to \cite{GaspersS15}.
Since we use the \emph{Separate, Measure and Conquer} method when the average degree drops to at most $\nicefrac{8}{3}$, we slightly generalize the separation computation from \cite{GaspersS15}, where the bound on the size of the separator depended only on the maximum degree.
A separation $(L,S,R)$ of a graph $G$ is a partition of the vertex set of $G$ such that every path from a vertex in $L$ to a vertex in $R$ contains a vertex from $S$.

\begin{restatable}{lemma}{lemSep}
\label{lem:sep}
  Let $B\in \mathbb{R}$.
  Let $\mu$ be a measure for graph problems such that
  for every graph $G=(V,E)$, every $R\subseteq V$, and every $v\in V$, we have that
  $|\mu(R\cup \{v\}) - \mu(R)| \le B$.
  Assume that $\mu(R)$, the restriction of $\mu$ to $R$, can be computed in polynomial time.
  If there is an algorithm computing a path decomposition of width at most $k$ of a graph $G$ in polynomial time,
  then there is a polynomial time algorithm computing a separation $(L,S,R)$ of $G$ with $|S|\le k$ and $|\mu(L)-\mu(R)| \le B$.
\end{restatable}
\newcommand{\proofLemSep}{%
  \begin{proof}
   The proof is basically the same as for the separation computation from \cite{GaspersS15}, but we repeat it here for completeness.
   First, compute a path decomposition of width $k$ in polynomial time.
   We view a path decomposition as a sequence of bags $(B_1, \dots, B_b)$ which are subsets of vertices such that for each edge of $G$, there is a bag containing both endpoints, and for each vertex of $G$, the bags containing this vertex form a non-empty consecutive subsequence. The width of a path decomposition is the maximum bag size minus one.
   We may assume that every two consecutive bags $B_i$, $B_{i+1}$ differ by exactly one vertex, otherwise we insert between $B_i$ and $B_{i+1}$ a sequence of bags where the vertices from $B_i \setminus B_{i+1}$ are removed one by one followed by a sequence of bags where the vertices of $B_{i+1} \setminus B_i$ are added one by one; this is the standard way to transform a path decomposition into a \emph{nice} path decomposition of the same width where the number of bags is polynomial in the number of vertices \cite{BodlaenderK96}.
   Note that each bag is a separator and a bag $B_i$ defines the separation $(L_i, B_i, R_i)$ with $L_i = (\bigcup_{j=1}^{i-1} B_j)\setminus B_i$ and $R_i = V \setminus (L_i\cup B_i)$.
   Since the first of these separations has $L_1=\emptyset$ and the last one has $R_b=\emptyset$, at least one of these separations has $|\mu_r(L_i)-\mu_r(R_i)| \le B$.
   Finding such a bag can clearly be done in polynomial time.
  \end{proof}
}
\longversion{\proofLemSep}

\noindent
We will use the lemma for graphs with maximum degree 3 and graphs with maximum degree 3 and average degree at most $\nicefrac{8}{3}$, for which path decompositions of width at most $\nicefrac{n}{6}+o(n)$ and $\nicefrac{n}{9}+o(n)$ can be computed in polynomial time, respectively \cite{fomin2009two,FominH06}.

One disadvantage of using the Separate, Measure and Conquer method for $A_{8/3}$ is that the algorithm needs to choose vertices for branching so that the size of the separator decreases in each branch.
However, Wahlstr\"{o}m's algorithm defers to branch on degree-3 vertices with all neighbors of degree 2 until this is no longer possible, since this case leads to the largest branching factor for degree 3.
For our approach, we instead rearrange the separator in some cases until we are only left with spider vertices, a structure where our algorithm cannot avoid branching on a degree-3 vertex with all neighbors of degree 2, we give a special weight to these spider vertices and upper bound their number.

\subparagraph{Potentials.}
To optimize the running time further, we also use potentials; see \cite{GaspersS12,Iwata11,Wahlstrom04}. These are constant weights that are added to the measure if certain global properties of the instance hold. For instance, we may use them to slightly increase the measure when an unfavorable branching strategy needs to be used. The constraint \eqref{eq:mu} for this unfavorable case then becomes less constraining, while all branchings that can lead to this unfavorable case get tighter constraints. This allows then to amortize unfavorable cases with favorable ones.

\section{Algorithm}

We first introduce notation necessary to present the algorithm.
Let $V(G)$ and $E(G)$ denote the vertex set and the edge set of the input graph $G$.
For a vertex $v \in V(G)$, its neighborhood, $N_G(v)$, is the set of vertices adjacent to $v$.
The \emph{closed neighborhood} of a vertex $v$ is $N_G[v] = N_G(v) \cup \{v\}$.
If $G$ is clear from context, we just use $N(v)$ and $N[v]$.

The degree of $v$ is denoted $d(v) = |N_G(v)|$.
An edge $uv \in E(G)$ is adjacent to vertex $u \in V(G)$ and $v \in V(G)$.
For two vertices $u$ and $v$ connected by a path, let $P \subset V(G)$ with $u,v \not\in P$ be the intermediate vertices between $u$ and $v$ on the path. If $P$ consists only of degree-2 vertices then we call $P$ a \emph{2-path} of $u$ and $v$.

The maximum degree of $G$ is denoted $\Delta(G)$ and $d(G) = 2 |E(G)| / |V(G)|$ is its \emph{average degree}.
A \emph{cubic} graph consists only of degree-3 vertices. A \emph{subcubic} graph has maximum degree at most 3.
A \emph{$(k_1, k_2,...,k_d)$ vertex} is a degree-$d$ vertex with all neighbors of degree $k_1, k_2, ... ,k_d$.
A separation $(L,S,R)$ of $G$ is a partition of its vertex set into the three sets $L,S,R$ such that no vertex in $L$ is adjacent to any vertex in $R$. The sets $L,S,R$ are also known as the \emph{left set}, \emph{separator}, and \emph{right set}.
Using a similar notion to \cite{GaspersS15}, a separation $(L,S,R)$ of $G$ is \emph{balanced} with respect to some measure $\mu$, and a branching constant $B$ if 
$|\mu(R) - \mu(L)| \leq 2B$ and \emph{imbalanced} if $|\mu(R) - \mu(L)| > 2B$.

By convention, $\mu(R) \geq \mu(L)$ otherwise, we swap $L$ and $R$. We use the measure $\mu_r$ defined on page \pageref{eq:deg3measure} to compute the separation in our algorithm.
We will now describe the algorithm \texttt{\#IS} which takes as input a graph $G$, a separation $(L,S,R)$, and a cardinality function $\mathbf{c} : \{0,1\} \times V(G) \to \mathbb{N}$, and computes the number of independent sets of $G$ weighted by the cardinality function $\mathbf{c}$. For clarity, let $\mathbf{c}_{out}(v) = \mathbf{c}(0,v)$ and $\mathbf{c}_{in}(v) = \mathbf{c}(1,v)$.
More precisely, it computes
\begin{align*}
ind(G,\mathbf{c}) = \sum_{X\subseteq V(G)} \mathds{1} (X \text{ is an independent set in } G) \cdot \prod_{v\in X} \mathbf{c}_{in}(v) \cdot \prod_{v\in V\setminus X} \mathbf{c}_{out}(v)
\end{align*}
where $\mathds{1}(\cdot)$ is an indicator function which returns 1 if its arguments is true and 0 otherwise.
Note that for a cardinality function $\mathbf{c}$ initialized to $\mathbf{c}(0,v)=\mathbf{c}(1,v)=1$ for every vertex $v\in V(G)$, we have that $ind(G, \mathbf{c})$ is the number of independent sets of $G$.
Cardinality functions are used for bookkeeping during the branching process and have been used in this line of work before.
The separation $(L,S,R)$ is initialized to $(\emptyset, \emptyset, V(G))$ and will only come into play when $G$ is subcubic and has no (3,3,3)-vertex.
In this case, the algorithm calls a subroutine \texttt{\#3IS}, which constitutes the main contribution of this paper.
\texttt{\#3IS} computes a balanced separation of $G$, preferring to branch on vertices in the separator, readjusting the separator as needed, and is analyzed using the Separate, Measure and Conquer method.

\subparagraph{Skeleton Graph.}
The skeleton graph $\Gamma(G)$, or just $\Gamma$, of a subcubic graph $G$ is a graph where the degree-3 vertices of $G$ are in bijection with the vertices in $\Gamma$. Two vertices in $\Gamma$ are adjacent if the corresponding vertices are adjacent in $G$, or there exists a 2-path between the corresponding vertices in $G$. If $G$ has a separation $(L,S,R)$ then denote $(L_{\Gamma}, S_{\Gamma}, R_{\Gamma})$ to be the same separation of $G$ in $\Gamma$ consisting of only degree-3 vertices. \emph{Dragging} refers to moving vertices or a set of vertices of $G$ from one component of $(L,S,R)$ to another, creating a new separation $(L', S', R')$ such that $S'$ is still a separator of $G$.

\subparagraph{Spider Vertices.}
As Wahlstr\"{o}m's \cite{Wahlstrom08} analysis showed, an unfavorable branching case occurs on vertices of degree 3 which have neighbors of degree (2,2,2). Due to our algorithm's handling of these vertices we narrowed down the undesirable vertices called \emph{spider vertices} down to a specific list of properties. If $s$ is a spider vertex then:
\begin{itemize}
  \item $s \in S$
  \item $s$ has neighbors of degree (2,2,2)
  \item Either:
  \begin{itemize}
     \item $|N_{\Gamma}(s) \cap L| = 2$ and $N_{\Gamma}(s) \cap R = \{r\}$ with $r$ having neighbors of degree (2,2,2). In this case we call $s$ a \emph{left spider vertex}
     \item $|N_{\Gamma}(s) \cap R| = 2$ and $N_{\Gamma}(s) \cap L = \{l\}$ with $l$ having neighbors of degree (2,2,2). In this case we call $s$ a \emph{right spider vertex}
     \item $|N_{\Gamma}(s) \cap L| = 1$, $|N_{\Gamma}(s) \cap R| = 1$, $N_{\Gamma}(s) \cap S = \{s'\}$ and $s'$ has neighbors of degree (2,2,2). In this case we call both $s$ and $s'$ a \emph{center spider vertex}, which occur in pairs.
   \end{itemize} 
\end{itemize}
A left spider vertex $s \in S$ can be dragged to the left along with the 2-path from $s$ to $r$. If this were ever to occur, then $r$ would be a right spider vertex, and vice versa.

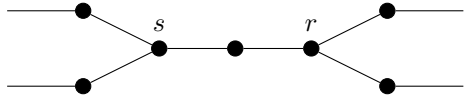
\begin{figure}[ht]
    \centering
    \begin{tikzpicture}[yscale=0.5]
    \node[vertex] (1) at (0,0) {};
    \node[vertex,label=above:$r$] (2) at (1,0) {};
    \node[vertex] (3) at (2,-1) {};
    \node[vertex] (4) at (2,1) {};
    \node[vertex,label=above:$s$] (7) at (-1,0) {};
    \node[vertex] (8) at (-2,1) {};
    \node[vertex] (9) at (-2,-1) {};    
    \draw (1)--(2)--(3)--(3,-1) (2)--(4)--(3,1) (1)--(7) 
    (-3,-1)--(9)--(7)--(8)--(-3,1);
    \end{tikzpicture}
    \caption{A left spider vertex $s$.}
\end{figure}

\subparagraph{Multiplier Reduction.}
We use a reduction called multiplier reduction to simplify graphs that have a cut vertex efficiently. Suppose $G$ has a separation $(V_1, \{x\}, V_2)$ and $G_1 = G[V_1 \cup \{x\}]$ has measure at most a constant $B$. The \emph{multiplier reduction} can be applied to compute \texttt{\#IS}($G, (L,S,R), \mathbf{c}$) as follows.
\begin{enumerate}
    \item Let:
    \begin{itemize}
        \item $G_{\text{out}} = G_1 \bs \{x\}$
        \item $G_{\text{in}} = G_1 \bs N_{G_1}[x]$
        \item $c_{\text{out}} = $\texttt{\#IS}($G_{\text{out}}$,($L[G_{\text{out}}],S[G_{\text{out}}],R[G_{\text{out}}]$),\textbf{c})
        \item $c_{\text{in}} = $ \texttt{\#IS}($G_{\text{in}}$, ($L[G_{\text{in}}],S[G_{\text{in}}],R[G_{\text{in}}]$), \textbf{c})
    \end{itemize}
    \item Modify \textbf{c} such that $\mathbf{c}_{\text{in}}(x) = \mathbf{c}_{\text{in}}(x) \cdot c_{\text{in}}$ and $\mathbf{c}_{\text{out}}(x) = \mathbf{c}_{\text{out}}(x) \cdot c_{\text{out}}$
    \item Return \texttt{\#IS}($G[V_2 \cup \{x\}]$, $(L,S,R)$, \textbf{c}) 
\end{enumerate}
Since $G_1$ has a measure of constant size, both steps 1 and 2 take polynomial time.


\subparagraph{Lazy 2-separator.} Suppose there is a vertex $x$ initially chosen to branch on as well as two vertices $\{y,z\} \subset V(G)$ with $d(y) \geq 3$ and $d(z) \geq 3$ such that $\{y,z\}$ is a separator which separates $x$ from $G$ in a constant measure subgraph. We call such vertices \emph{lazy 2-separators}, for a vertex $x$.
Similar to Walhstr\"{o}m's elimination of separators of size 2 in \cite{Wahlstrom04}, in line \ref{algoln:is-sep2}
of \texttt{\#IS} instead of branching on $x$, if there exists a lazy 2-separator $\{y,z\}$ for $x$ we branch on $y$. 
A multiplier reduction will be performed on $z$ in the recursive calls. 
Prioritizing \emph{lazy 2-separators} allows to exclude some unfavorable cases when branching on $x$.


\subparagraph{Associated Average Degree.}
Similar to \cite{Wahlstrom08}, we define the \emph{associated average degree} of a vertex $x \in V(G)$ as $\alpha(x) / \beta(x)$, in $G$ with average degree $d(G) = k$ where
\begin{equation}
    \alpha(x) = d(x) + | \{ y \in N(x) : d(y) < k\} |, \text{ and }
    \beta(x)  = 1 + \sum_{\hspace{-5mm} \{ y \in N(x) | d(y) < k\} \hspace{-8mm}} 1/d(y).
\end{equation}
By selecting vertices with high associated average degree, our algorithm prioritizes branching on vertices with larger decreases in measure.

\subparagraph{Branching.}
We now outline the branching routine used to recursively solve smaller instances of the problem. Suppose we have a graph $G$, a separation $(L,S,R)$, and a cardinality function \textbf{c}. For a vertex $x$ we denote the following steps as \emph{branching on $x$}. 
\begin{enumerate}
    \item Let:
    \begin{itemize}
        \item $G_{\text{out}} = G \bs \{x\}$ 
        \item $G_{\text{in}} = G \bs (N(x) \cup \{x\})$
        \item $c_{\text{out}} = $ \texttt{\#IS}($G_{\text{out}}$, ($L[G_{\text{out}}],S[G_{\text{out}}],R[G_{\text{out}}]$), \textbf{c})
        \item $c_{\text{in}} = $ \texttt{\#IS}($G_{\text{in}}$, ($L[G_{\text{in}}],S[G_{\text{in}}],R[G_{\text{in}}]$), \textbf{c})
        \item $c'_{\text{out}} = \mathbf{c}_{\text{out}}(x)$
        \item $c'_{\text{in}} = \mathbf{c}_{\text{in}}(x) \cdot \prod_{v \in N(x)} \mathbf{c}_{\text{out}}(v)$
    \end{itemize}
    \item Return $c'_{\text{out}}\cdot c_{\text{out}} + c'_{\text{in}}\cdot c_{\text{in}}$
\end{enumerate}


\begin{algorithm}[t]
    \SetKwInOut{Algorithm}{Algorithm}
    \SetKwInOut{Input}{Input}
    \SetKwInOut{Output}{Output}
    \Algorithm{\texttt{\#IS}($G,(L,S,R), \mathbf{c}$) - \textsc{\#Independent Set} algorithm}
    \Input{Graph $G = (V,E)$, separation $(L,S,R)$ of $G$, cardinality function $\mathbf{c}$}
    \Output{$ind(G,\mathbf{c})$}
    
    \If{$V = \emptyset$}{
        \Return 1 
        \label{algoln:is-empty}
    }
    \If{$|V| = 1$}{
        \Return $\mathbf{c}_{\text{in}}(x) + \mathbf{c}_{\text{out}}(x)$ where $V=\{x\}$
        \label{algoln:is-single}
    }
    \If{$\Delta(G) \leq 2$}{
        \Return a solution in polynomial time
        \label{algoln:is-deg2}
    }
    \ElseIf{$G$ is not connected and has $j$ connected components $G_1, G_2,...,G_j$}{
        \Return $\prod_{i=1}^j$ \texttt{\#IS}$(G_i,(\emptyset,\emptyset,V(G_i)), \mathbf{c})$
        \label{algoln:is-conn-comp}
    }
    \ElseIf{$\Delta(G) = 4$, and all degree-4 vertices of $G$ only have degree-2 neighbors and there exists a vertex $x$ where $d(x) = 4$ and $x$ has a 2-path to a degree-3 vertex}{
    	Branch on $x$
    }
    \Else{Let vertex $x \in V$ be a vertex of maximum degree, secondarily maximizing the associated average degree $\alpha(x) / \beta(x)$ \newline
        \If{the multiplier reduction applies}{
            Apply the multiplier reduction.
            \label{algoln:is-mult-red}
        }
        \ElseIf{there exists a separator of size 2: $\{y, z\}$, with $d(y) \geq 3$ and $d(z) \geq 3$ whose removal leaves $G$ disconnected and either removes or leaves $N_G[x]$ in a component with constant measure at most $B$}{
            Branch on $y$. 
            \label{algoln:is-sep2}
        }
        \Else{
            \If{$\Delta(G) = 3$ and $G$ has no (3,3,3) vertex}{
              \Return \texttt{\#3IS}($G,(L,S,R), \mathbf{c}$)
              \label{algoln:is-3is}
            }

            \Else{
              Branch on $x$
              \label{algoln:is-branch}
            }
        }
    }
    \end{algorithm} 

\begin{algorithm}[tb]
    \SetKwInOut{Algorithm}{Algorithm}
    \SetKwInOut{Input}{Input}
    \SetKwInOut{Output}{Output}
    \Algorithm{\texttt{\#3IS}($G$,$(L,S,R)$, $\mathbf{c}$) - \textsc{\#Independent Set} algorithm for subcubic graphs with no (3,3,3) vertex}
    \Input{Graph $G = (V,E)$, separation $(L,S,R)$ of $G$, cardinality function $\mathbf{c}$}
    \Output{$ind(G,\mathbf{c})$}
    \If{$S = \emptyset$}{
        Compute a balanced separation $(L,S,R)$ with respect to the measure $\mu$ using Lemma \ref{lem:sep}.
        \label{3is:new-sep}
    }
    \If{$\mu_r(L) > \mu_r(R)$}{
        Swap $L$ and $R$
        \label{3is:swapLR}
    }
    $(L,S,R)$ := \texttt{simplify}$(G, (L, S, R))$\\
    \label{3is:simplifycall}
    Let $s \in S$ be a maximum degree vertex with maximum associated average degree\\
    \If{the multiplier reduction applies}{
        Apply the multiplier reduction.
        \label{3is:multiplierreduction}
    }
    \ElseIf{there exists a separator of size 2: $\{y, z\}$, with $d(y) \geq 3$ and $d(z) \geq 3$ whose removal leaves $G$ disconnected and either removes or leaves $N_{G}[s]$ in a component with constant measure at most $B$}{
        Branch on $y$.\\
        \label{3is:sep2}
    }
    \ElseIf{$\mu_r(R) - \mu_r(L) \leq 2B$ and $s$ has neighbors of degree (2,2,2)}{
        \Return \texttt{spider}$(s, G, (L,S,R), \mathbf{c})$
        \label{3is:spidercall}
    }
    \ElseIf{$\mu_r(R) - \mu_r(L) > 2B$ and $s$ has two neighbors in $L$ and one neighbor $r$ in $R$, let $r'$ be the first degree-3 vertex or vertex from $S$ encountered when moving from $s$ to the right along a 2-path in $G$}{
        \Return \texttt{\#IS}($G$, $(L \cup P \cup \{s\}, (S\bs\{s\}) \cup \{r'\}, R \bs (P \cup \{r'\}))$, \textbf{c})
        \label{3is:imbaldrag-tobranch}
    }
    \ElseIf{$\mu_r(R) - \mu_r(L) > 2B$ and there exists $r \in N_{\Gamma}(s) \cap R$ with $N_{\Gamma}(r) \cap R = \emptyset$}{
        Let $\{r, r'\} = N_{\Gamma}(s) \cap R$ with $N_{\Gamma}(r) \cap R = \emptyset$. \\
        Branch on $r'$.
        \label{3is:bad_imbal_case}
    }
    \Else{
        Branch on $s$.\\
        \label{3is:branch}
    }
\end{algorithm}

\begin{algorithm}[!ht]
\SetKwInOut{Algorithm}{Algorithm}
\SetKwInOut{Input}{Input}
\SetKwInOut{Output}{Output}
\Algorithm{\texttt{simplify}($G,(L,S,R)$) - Applies simplification rules.}
\Input{Graph $G = (V,E)$, separation $(L,S,R)$ of $G$}
\Output{$(L,S,R)$}
\If{there exists a vertex $s \in S$ with no neighbor in $L$}{
    \Return \texttt{simplify}($L,S \backslash \{s\}, R \cup \{s\}$)
    \label{simplify:no-neg-L}
}
\ElseIf{there exists a vertex $s \in S$ with no neighbor in $R$}{
    \Return \texttt{simplify}($L \cup \{s\},S \backslash \{s\}, R$)
    \label{simplify:no-neg-R}
}
\ElseIf{there exists a vertex $s \in S$ with $d(s) = 2$}{
    \If{$(L,S,R)$ is balanced}{
        Let $l \in (N_{\Gamma} \cap L) \cup S$. Let $P$ be the 2-path for $s$ and $l$. \\
        \Return \texttt{simplify}($G$, $(L \bs (P \cup \{l\}), (S\bs\{s\}) \cup \{l\}, R \cup P \cup \{s\})$) 
        \label{simplify:bal-drag}
    }
    \Else{
        Let $r \in (N_{\Gamma} \cap R) \cup S$. Let $P$ be the 2-path for $s$ and $r$. \\
        \Return \texttt{simplify}($G$, $(L \cup P \cup \{s\}, (S\bs\{s\}) \cup \{r\}, R \bs (P \cup \{r\}))$)    
        \label{simplify:imbal-drag}
    }
}
\ElseIf{there exists a vertex $s \in S$ which does not have a vertex $l \in N_{\Gamma}(s) \cap L$ such that $N_{\Gamma}(l) \cap L \neq \emptyset$}{
    For $l \in (N_{\Gamma}(s) \cap L)$, let $A_l$ = $\left(N_{\Gamma}(l) \cap S\right)$, let $P_{(s,l)}$ be the 2-path from $s$ to $l$, and for $a \in A_l$ let $P_{(l,a)} \subset V(G)$ be the 2-path from $l$ to $a$. \\
    Let $B = (N_{\Gamma}(s) \cap S)$ and for $b \in B$ let $Q_b \subset V(G)$ be the 2-path from $s$ to $b$. \\
    Let $C = \left( \bigcup_{b \in B} Q_b \right) \cup (N_{\Gamma}(s) \cap L) \cup \left( \bigcup_{l \in N_{\Gamma}(s) \cap L} A_l \cup P_{(s,l)} \cup \left( \bigcup_{a \in A_l} P_{(l,a)} \right) \right)$ \\
    \Return \texttt{simplify}($G,(L \bs C,S \backslash (\{s\} \cup C), R \cup \{s\} \cup C)$)
    \label{simplify:no-l-skeleton}
}
\ElseIf{there exists a vertex $s \in S$ which does not have a vertex $r \in N_G(s) \cap R$ such that $N_{\Gamma}(r) \cap R \neq \emptyset$}{
    For $r \in (N_{\Gamma}(s) \cap R)$, let $A_r$ = $\left(N_{\Gamma}(r) \cap S\right)$, let $P_{(s,r)}$ be the 2-path from $s$ to $r$, and for $a \in A_r$ let $P_{(r,a)} \subset V(G)$ be the 2-path from $r$ to $a$. \\
    Let $B = (N_{\Gamma}(s) \cap S)$ and for $b \in B$ let $Q_b \subset V(G)$ be the 2-path from $s$ to $b$. \\
    Let $C = \left( \bigcup_{b \in B} Q_b \right) \cup (N_{\Gamma}(s) \cap R) \cup \left( \bigcup_{r \in N_{\Gamma}(s) \cap R} A_r \cup P_{(s,l)} \cup \left( \bigcup_{a \in A_r} P_{(r,a)} \right) \right)$ \\
    \Return \texttt{simplify}($G,(L \cup\{s\} \cup C,S \backslash (\{s\} \cup C), R \bs C)$)
    \label{simplify:no-r-skeleton}
}

\Else{
  \Return($G,(L, S, R)$)
  \label{simplify:return}
}        
\end{algorithm} 

\begin{algorithm}[tb]
\SetKwInOut{Algorithm}{Algorithm}
\SetKwInOut{Input}{Input}
\SetKwInOut{Output}{Output}
\Algorithm{\texttt{spider}($s,G,(L,S,R), \mathbf{c}$) - Handles vertex $s$ with neighbor degree (2,2,2)}
\Input{Vertex $s$ with neighbors of degree (2,2,2), Graph $G = (V,E)$, separation $(L,S,R)$ of $G$, cardinality function $\mathbf{c}$}
\Output{$ind(G,\mathbf{c})$}

\If{$|N_{G}(s) \cap R| = 1$}{
    Let $\{r\} = N_{\Gamma}(s) \cap R$ and let $P_r$ be the 2-path from $s$ to $r$.\\
    \If{$r$ does not have neighbor degree (2,2,2)}{
      \Return \texttt{\#3IS}($G, (L \cup P_r,(S \cup \{r\}) \bs \{s\},R \bs (P_r \cup \{r\})), \mathbf{c}$)
      \label{spider:pull-r}
    }
}
\ElseIf{$|N_{G}(s) \cap L| = 1$}{
    Let $\{l\} = N_{\Gamma}(s) \cap L$ and let $P_l$ be the 2-path from $s$ to $l$.\\
    \If{$l$ does not have neighbor degree (2,2,2)}{
      \Return \texttt{\#3IS}($G, (L \bs (P_l \cup \{l\}),(S \cup \{l\}) \bs \{s\},R \cup P_l ), \mathbf{c}$)
      \label{spider:pull-l}
    }
}
\ElseIf{$|N_{\Gamma}(s) \cap S_{\Gamma}| = 1$}{
    Let $\{s'\} = N_{\Gamma}(s) \cap S_{\Gamma}$, $\{l\} = N_{\Gamma}(s) \cap L_{\Gamma}$ and $\{s, s_1, s_2\} = N_{\Gamma}(l)$.\\
    \For{$i \in \{1,2\}$}{
        \If{$s_i \in S$ and $|N_{G}(s_i) \cap R| = 1$}{
            Let $\{r_i\} = N_{G}(s_i) \cap R$\\
            $(L,S,R) := (L \cup \{s_i\}, (S \cup \{r_i\}) \bs \{s_i\} , R \bs \{r_i\})$ \\
            \label{spider:refineLSR-spider-drag}
        }
    }
    Branch on $l$.
    \label{spider:branch-l}
}
\Else{
    Branch on $s$.
    \label{spider:branch-s}
}


\end{algorithm}

\section{Running Time Analysis}
This section describes the running time analysis for \texttt{\#IS} and \texttt{\#3IS}, conducted via compound measures. Constraints are presented as branching vectors $(\delta_1, \delta_2)$ which equates to the constraints $2^{-\delta_1} + 2^{-\delta_2} \leq 1$.  We first describe some special vertex weights.


\subsection{Measures}
\label{subsec:measure}
\subparagraph{Measure with no (3,3,3) vertex.}

When using the Separate, Measure and Conquer technique from \cite{GaspersS15} the measure of a cubic graph instance $G$ with no (3,3,3) vertices consists of additive components $\mu_s$ and $\mu_r$, the measure of vertices in the separator, and those in either $L$ or $R$, respectively. Let $S' \subseteq S$ be the set of all spider vertices, $s_i$ and $r_i$ refer to the weight attributed to a separator vertex and a right vertex, in $R$ or $L$, respectively, of degree $i$. Left and right spider vertices have weight $s_3'$. In a center spider vertex pair $s$ and $s'$, one of them has weight $s_3'$ while the other takes on an ordinary weight of $s_3$. These structurally applied weights allows amortization of the spider vertex cases against non-spider vertices. Define the measure $\mu_{8/3}$ as
\begin{equation}
\label{eq:deg3measure}
\mu_{{8}/{3}} = \mu_s(S) + \mu_r(R) + \mu_o(L,S,R),
\end{equation}
where $\mu_s(S) = |S'| \cdot s_3' + \sum_{v \in S\bs S'} s_{d(v)}$, $\mu_r(R) = \sum_{v \in R} r_{d(v)}$, $B = 6 s_3$ and 
\[\mu_o(L,S,R) = \max\left\{0, B - \frac{\mu_r(R) - \mu_r(L)}{2}\right\} + (1+B) \cdot \log_{1 + \epsilon}(\mu_r(R) + \mu_s(S)).\]
We also require that $s_i \geq s_{i-1}$ and $r_i \geq r_{i-1}$ for $i \in \{1,2,3\}$.
The constant $B$ is larger than the maximum change in imbalance in each transformation in the analysis, except the separation transformation. 

\begin{restatable}{lemma}{lemUpperBound}\label{lem:upperBound}
For a balanced separation $(L,S,R)$ of a graph $G$ with average degree $d = d(G)$, an upper bound for the measure $\mu_{8/3}$ is:
\[\mu_{8/3}(d) \leq \begin{cases}
  \frac{n}{6} (d-2) s_3' + \frac{1}{2}\left(\frac{5n}{6} (d-2) r_3 + n(3-d) r_2\right) \\
  + \mu_o(L,S,R) + o(n)  & \text{if } 2 \leq d \leq \frac{28}{11}
   \\
  \frac{n}{4} (8 - 3d) s_3' + \frac{n}{12} (11 d - 28) s_3 + \frac{1}{2}\left(\frac{5n}{6} (d-2) r_3 + n(3-d) r_2 \right) \\
   + \mu_o(L,S,R) + o(n) & \text{if } \frac{28}{11} < d \leq \frac{8}{3} \\
\end{cases}\]
which is maximised when $d = \frac{8}{3}$ with the value
\[\mu_{8/3} \leq \frac{n}{9}s_3 + \frac{1}{2}\left(\frac{5n}{9}r_3 + \frac{n}{3}r_2 \right) + \mu_o(L,S,R) + o(n)\]
if constraints $\frac{r_2}{2} \leq \frac{s_3'}{11} + \frac{5r_3}{22} + \frac{5r_2}{22} \leq \frac{s_3}{9}+ \frac{5r_3}{18} +  \frac{r_2}{3}$ are satisfied.
\end{restatable}

\newcommand{\proofLemUpperBound}{
\begin{proof}
Let $d = d(G)$ be the average degree of $G$. For an appropriate upper bound of $\mu_{8/3}$ we first consider the upper bound on the number of separator vertices, also giving us an upper bound on the number of spider vertices:
\begin{equation}
\label{eq:spiderub1}
\text{\#Spiders} \leq |S| \leq \frac{n_3}{6} + o(n_3) = \frac{n(d-2)}{6} + o(n)
\end{equation}
where $n_3 = n(d - 2)$ is the number of degree-3 vertices in $G$, since a subcubic graph with $n_3$ vertices of degree 3 has pathwidth at most $\frac{n_3}{6} + o(n_3)$ \cite{fomin2009two}

As we have no vertex with neighbors of (3,3,3), every degree-3 vertex is incident to an edge incident to a degree-2 vertex. However, each spider vertex has need 4 more edges incident to degree-2 vertices. As the number of edges incident to degree-2 vertices is $2n_2$ where $n_2 = n(3 - d)$ is the number of degree-2 vertices in $G$, and there are at least $n_3$ of those edges taken up to be incident to a degree-3 vertex, then an upper bound on the number of spiders is:
\begin{equation}
\label{eq:spiderub2}
\text{\#Spiders} \leq \frac{2n_2 - n_3}{4} = n\left( 2 - \frac{3}{4}d \right)
\end{equation}

Since both upper bounds are valid for all $2 \leq d \leq 8/3$ then a more accurate upper bound can be found by taking the minimum of Equation \ref{eq:spiderub1} and \ref{eq:spiderub2}. This results in:
\[\text{\#Spiders} \leq \begin{cases}
  \frac{n}{6} (d-2) + o(n) & \text{ if } 2 \leq d \leq \frac{28}{11} \\
  n\left( 2 - \frac{3}{4}d \right) & \text{ if } \frac{28}{11} < d \leq \frac{8}{3} \\
\end{cases}\]

As $|S| \leq \frac{n}{6} (d-2)$ for all $2 \leq d \leq \frac{8}{3}$, with the weight for spider vertices $s_3$ being greater than regular non-spider degree-3 vertices in the separator, then an upper bound for $\mu_{8/3}$ would have as many spider vertices in $S$ as possible for a given average degree $d$. For $2 \leq d \leq \frac{28}{11}$ it is possible to have all vertices in $S$ be spider vertices, so this gives the greatest value of $\mu_{8/3}$. However, from $\frac{28}{11} < d \leq \frac{8}{3}$ we use Equation \ref{eq:spiderub1} to upper bound $|S|$ and also place in $S$ as many spider vertices with weight $s_3'$ as Equation \ref{eq:spiderub2} allows, with the rest of the vertices in $S$ being of weight $s_3$.
\[\mu_{8/3} \leq \begin{cases}
  \frac{n}{6} (d-2) s_3' + \frac{1}{2}\left(\frac{5n}{6} (d-2) r_3 + n(3-d) r_2\right) \\
  + \mu_o(L,S,R) + o(n)  & \text{if } 2 \leq d \leq \frac{28}{11}
   \\
  \frac{n}{4} (8 - 3d) s_3' + \frac{n}{12} (11 d - 28) s_3 + \frac{1}{2}\left(\frac{5n}{6} (d-2) r_3 + n(3-d) r_2 \right) \\
   + \mu_o(L,S,R) + o(n) & \text{if } \frac{28}{11} < d \leq \frac{8}{3} \\
\end{cases}\]
For maximum value, let $f_1(d) = \frac{n}{6} (d-2) s_3' + \frac{1}{2}(\frac{5n}{6} (d-2) r_3 + n(3-d) r_2)$ and $f_2(d) = \frac{n}{4} (8 - 3d) s_3' + \frac{n}{12} (11 d - 28) s_3 + \frac{1}{2}(\frac{5n}{6} (d-2) r_3 + n(3-d) r_2)$. We notice that $f_1$ and $f_2$ are both linear functions in $d$ and $f_1(\frac{28}{11}) = f_2(\frac{28}{11})$ meaning that the endpoints: $f_1(2), f_2\left(\frac{28}{11}\right), f_2\left(\frac{8}{3}\right)$ are the only points of interest. For the measure to not increase on lower degrees, we require that $f_1(2) \leq  f_2\left(\frac{28}{11}\right) \leq f_2\left(\frac{8}{3}\right)$ which results in the constraints
\[\frac{r_2}{2} \leq \frac{s_3'}{11} + \frac{5r_3}{22} + \frac{5r_2}{22} \leq \frac{s_3}{9} + \frac{5r_3}{18} +  \frac{r_2}{3}\]
and the maximum value achieved by $f_2$ when average degree $d = \frac{8}{3}$:
\[\mu_{8/3} \leq f_2\left(\frac{8}{3}\right) + \mu_o(L,S,R) + o(n) = \frac{n}{9}s_3 + \frac{1}{2}\left(\frac{5n}{9}r_3 + \frac{n}{3}r_2 \right) + \mu_o(L,S,R) + o(n)\]


\end{proof}
}

\subparagraph{General Measure. } In order to analyze higher degree cases, we use a measure of the form
\[
    \mu_i(G) = \sum_{v \in G} r_{d(v)} + \mu_o(L,S,R) \quad \text{ where } \Delta(G) = i
\]
for each part of the compound measure. The term $\mu_o(L,S,R)$ is the same sub-linear term from the Separate, Measure and Conquer analysis on cubic graphs which needs to be propagated into the higher degree analyses.

\subsection{Degree 3 Analysis}
\texttt{\#IS} can be solved in polynomial time when $\Delta(G) \le 2$ \cite{roth1996hardness}.
However, stepping up to cubic graphs is a much harder problem. Greenhill \cite{greenhill2000complexity} proves that counting independent sets is actually a \#P-hard problem even for graphs with maximum degree 3.

\begin{restatable}{lemma}{lemEightThirdRunningTime}\label{lem:eightThirdRunningTime}
    Algorithm \texttt{\#IS} applied to a graph $G$ with $\Delta(G) \leq 3$ and no $(3,3,3)$ vertex has running time $O(\rteightthirds^n)$.
\end{restatable}

\shortversion{
\subparagraph{Proof sketch.}
We present a sketch of the proof, emphasizing the tight constraints generated from \texttt{\#3IS}, \texttt{simplify} and \texttt{spider}. A complete analysis will be deferred until the appendix. As suggested in \cite{GaspersS15}, each case will provide constraints that the weights described above will need to satisfy. 

Some trivial constraints we must satisfy are $r_0 = r_1 = s_0 = s_1 = 0$ since these vertices can easily be eliminated and require no branching rules. Our algorithm considers skeleton graph vertices, and several rules drag entire 2-paths from one separation to another, requiring $r_2 = 0$. In \texttt{simplify}, line \ref{simplify:bal-drag} implies the constraint $s_2 + s_3' + \frac{1}{2}(r_2 - r_3) \leq 0$, enabling us to move a degree-3 vertex into the separator by dragging out a degree-2 vertex.

From \texttt{\#3IS}, line \ref{3is:new-sep} imposes the constraint $\frac{1}{6}s_3' + \frac{5}{12}r_3 \leq r_3$.
If a (2,2,3) vertex $s$ is chosen to branch on in line \ref{3is:branch} as shown in Figure \ref{fig:non-spider-branching}(a), then we get the constraint $\left(s_3 + \Delta s_3 + \frac{1}{2}( 2\Delta r_3) - 3 \delta, s_3 + 2 \Delta s_3 + \frac{1}{2}(r_3 + 2 \Delta r_3 ) - 4 \delta \right)$. The last tight constraint is from \texttt{spider} line \ref{spider:branch-s}, displayed in Figure \ref{fig:spider-branching}(a), giving the constraint $\left(s_3' + \frac{3}{2}\Delta r_3, s_3' + \frac{3}{2}\Delta r_3 \right)$.

While the cases in Figure \ref{fig:non-spider-branching}(b) and Figure \ref{fig:spider-branching}(b) are not tight, they are of interest since these cases branch on vertices located outside the separator and it is guaranteed that $s$ is removed from the separator after branching.
}

 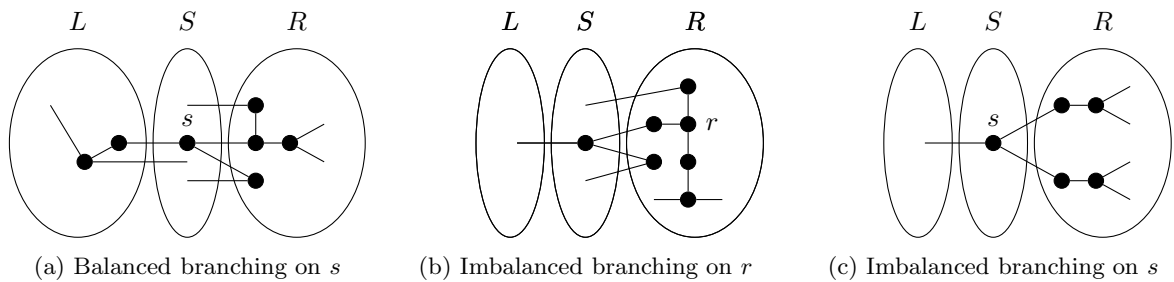
\begin{figure}[!htb]
    \centering
    \begin{minipage}{0.33\textwidth}
        \centering
        \begin{tikzpicture}[yscale=0.5, xscale=0.9]
            \node[vertex, label=above:$s$] (1) at (0,0) {};

            \node[vertex] (2) at (1,0) {};
            \node[vertex] (3) at (1,-1) {};
            \node[vertex] (4) at (1.5,0) {};
            \node[vertex] (5) at (1,1) {};

            \node[vertex] (6) at (-1,0) {};
            \node[vertex] (7) at (-1.5,-0.5) {};

            \draw (0,0) ellipse (0.5 and 2.5);
            \draw (1.6,0) ellipse (1 and 2.5);
            \draw (-1.6,0) ellipse (1 and 2.5);
            \node[draw=none, label=above:$S$] (slabel) at (0,2.5) {};
            \node[draw=none, label=above:$R$] (rlabel) at (1.6,2.5) {};
            \node[draw=none, label=above:$L$] (llabel) at (-1.6,2.5) {};
            \node[draw=none, label=below:\small (a) Balanced branching on $s$] (slabel) at (0,-2.5) {};
            
            \draw (0,1)--(5)--(2)--(1)--(3)--(0,-1) (2)--(4)--(2, 0.5) (4)--(2,-0.5); 
            \draw (1)--(6)--(7) (-2,1)--(7)--(0,-0.5);

            \end{tikzpicture}
    \end{minipage}%
    \begin{minipage}{0.33\textwidth}
        \centering
        \begin{tikzpicture}[yscale=0.5, xscale=0.9]
            \node[vertex] (1) at (0,0) {};

            \node[vertex] (2) at (1,0.5) {};
            \node[vertex, label=right:$r$] (3) at (1.5,0.5) {};
            \node[vertex] (4) at (1.5,1.5) {};
            \node[vertex] (5) at (1.5,-0.5) {};
            \node[vertex] (6) at (1.5,-1.5) {};
            \node[vertex] (7) at (1,-0.5) {};

            \draw (0,0) ellipse (0.5 and 2.5);
            \draw (1.6,0) ellipse (1 and 2.5);
            \draw (-1.1,0) ellipse (0.5 and 2.5);
            \node[draw=none, label=above:$S$] (slabel) at (0,2.5) {};
            \node[draw=none, label=above:$R$] (rlabel) at (1.6,2.5) {};
            \node[draw=none, label=above:$L$] (llabel) at (-1.1,2.5) {};
            
            \draw (4)--(3)--(2)--(1)--(7)--(0,-01) (3)--(5)--(6) (0,1)--(4) (1,-1.5)--(6)--(2,-1.5); 
            \draw (1)--(-1,-0);

            \draw (0,0) ellipse (0.5 and 2.5);
            \draw (1.6,0) ellipse (1 and 2.5);
            \draw (-1.1,0) ellipse (0.5 and 2.5);
            \node[draw=none, label=above:$S$] (slabel) at (0,2.5) {};
            \node[draw=none, label=above:$R$] (rlabel) at (1.6,2.5) {};
            \node[draw=none, label=above:$L$] (llabel) at (-1.1,2.5) {};
            \node[draw=none, label=below:\small (b) Imbalanced branching on $r$] (slabel) at (0,-2.5) {};
            
            \draw  ; 
            \draw (1)--(-1,0);
        \end{tikzpicture}
    \end{minipage}
    \begin{minipage}{.33\textwidth}
        \centering
        \begin{tikzpicture}[yscale=0.5, xscale=0.9]
            \node[vertex, label=above:$s$] (1) at (0,0) {};

            \node[vertex] (2) at (1,1) {};
            \node[vertex] (3) at (1,-1) {};
            \node[vertex] (4) at (1.5,1) {};
            \node[vertex] (5) at (1.5,-1) {};

            \draw (0,0) ellipse (0.5 and 2.5);
            \draw (1.6,0) ellipse (1 and 2.5);
            \draw (-1.1,0) ellipse (0.5 and 2.5);
            \node[draw=none, label=above:$S$] (slabel) at (0,2.5) {};
            \node[draw=none, label=above:$R$] (rlabel) at (1.6,2.5) {};
            \node[draw=none, label=above:$L$] (llabel) at (-1.1,2.5) {};
            \node[draw=none, label=below:\small (c) Imbalanced branching on $s$] (slabel) at (0,-2.5) {};
            
            \draw (4)--(2)--(1)--(3)--(5) (2,1.5)--(4)--(2,0.5) (2,-1.5)--(5)--(2,-0.5); 
            \draw (1)--(-1,-0);

        \end{tikzpicture}
    \end{minipage}%
    \caption{Worst case configurations for non-spider vertex branching in \texttt{\#3IS}}
    \label{fig:non-spider-branching}
\end{figure}

\begin{figure}[ht]
    \begin{minipage}{.5\textwidth}
        \centering
        \begin{tikzpicture}[yscale=0.5]
            \node[vertex, label=above:$s$] (1) at (0,0) {};

            \node[vertex] (2) at (1,1) {};
            \node[vertex] (3) at (1,-1) {};

            \node[vertex] (4) at (-1,0) {};
            \node[vertex, label=left:$l$] (5) at (-1.5,0) {};
            \node[vertex] (6) at (-1.5,1) {};
            \node[vertex] (7) at (-1.5,2) {};
            \node[vertex] (8) at (-1.5,-1) {};
            \node[vertex] (9) at (-1.5,-2) {};

            \draw (0,0) ellipse (0.5 and 2.5);
            \draw (1.2,0) ellipse (0.6 and 2.5);
            \draw (-1.5,0) ellipse (0.8 and 2.5);
            \node[draw=none, label=above:$S$] (slabel) at (0,2.5) {};
            \node[draw=none, label=above:$R$] (rlabel) at (1.2,2.5) {};
            \node[draw=none, label=above:$L$] (llabel) at (-1.6,2.5) {};
            \node[draw=none, label=below:\small (a) Balanced branching on $l$] (slabel) at (0,-2.5) {};
            
            \draw (2)--(1)--(3)--(0,-2) (1.5,1.5)--(2)--(1.5,0.5);
            \draw (9)--(8)--(5)--(6)--(7) (5)--(4)--(1) (-1,-1.5)--(9)--(-2,-1.5) (-1,1.5)--(7)--(-2,1.5);
        \end{tikzpicture}

    \end{minipage}%
    \begin{minipage}{.5\textwidth}
        \centering
        \begin{tikzpicture}[yscale=0.5]
            \node[vertex, label=above:$s$] (1) at (0,0) {};

            \node[vertex] (2) at (1,1) {};
            \node[vertex] (3) at (1,-1) {};
            \node[vertex] (4) at (1.5,1) {};
            \node[vertex] (5) at (1.5,-1) {};

            \node[vertex] (6) at (-1,0) {};
            \node[vertex] (7) at (-1.7,0) {};

            \draw (0,0) ellipse (0.5 and 2.5);
            \draw (1.6,0) ellipse (1 and 2.5);
            \draw (-1.6,0) ellipse (1 and 2.5);
            \node[draw=none, label=above:$S$] (slabel) at (0,2.5) {};
            \node[draw=none, label=above:$R$] (rlabel) at (1.6,2.5) {};
            \node[draw=none, label=above:$L$] (llabel) at (-1.6,2.5) {};
            \node[draw=none, label=below:\small (b) Balanced branching on $s$] (slabel) at (0,-2.5) {};
            
            \draw (4)--(2)--(1)--(3)--(5) (2,1.5)--(4)--(2,0.5) (2,-1.5)--(5)--(2,-0.5); 
            \draw (1)--(6)--(7) (-2,-0.5)--(7)--(-2,0.5);
        \end{tikzpicture}
    \end{minipage}

    \caption{Worst case configurations for spider vertex branching in \texttt{spider}}
    \label{fig:spider-branching}
\end{figure}

\newcommand{\proofLemEightThirdRunningTime}{
\subparagraph{Proof.}
We will analyze the running time with respect to the measure $\mu_{8/3}$ described above. 
As suggested in \cite{GaspersS15} we will provide constraints that these weights need to satisfy, and the provided values minimize the measure. 
The measure $\mu_{8/3}$ can be viewed in two regimes; a balanced separation, where $\mu_r(R) - \mu_r(L) \leq 2B$ resulting in $\mu_{8/3} = \mu_s(S) + \frac{1}{2}( \mu_r(R) - \mu_r(L)) + \mu_o(L,S,R)$ and an imbalanced separation, where $\mu_r(R) - \mu_r(L) > 2B$ resulting in $\mu_{8/3} = \mu_s(S) + \mu_r(R) + \mu_o(L,S,R)$. 
To characterize decreases in vertex degrees, let $\Delta s_i = s_i - s_{i-1}$ and $\Delta r_i = r_i - r_{i-1}$. Trivial constraints are 
\begin{equation*} \label{eq:trivial-cons}   
    \begin{aligned}
        r_0 = r_1 = 0 \quad\quad& s_0 = s_1 = 0.
    \end{aligned} 
\end{equation*}
Our algorithm handles 2-paths as if they were single edges. Therefore we constrain that $r_2 = 0$.

\subparagraph{Constraints from \texttt{\#IS}}
Simplification rules in lines \ref{algoln:is-empty} to \ref{algoln:is-conn-comp} in \texttt{\#IS} take polynomial time. If we are given a graph $G$ with $\Delta(G) \leq 3$ and no $(3,3,3)$ vertex and the lazy 2-separator rule in line \ref{algoln:is-sep2} did not apply, then we enter the subroutine \texttt{\#3IS}. 

\subparagraph{Constraints from \texttt{simplify}}
The simplification rules in \texttt{simplify} either reduce the separator size by removing a vertex or the rule drags degree-2 vertices in $S$ away making $S$ consist only of degree-3 vertices.
For vertex dragging to $R$ in line \ref{simplify:no-neg-L} of \texttt{simplify}, the most constraining instances are the balanced ones:
\begin{equation*}
    - s_d + r_d \leq 0 \text{ where } d \in \{2,3\} \text{ and } -s_3' + r_3 \leq 0 \enspace.
    \label{eq:vertexdragright}
\end{equation*}
However, for vertex dragging to $L$ in \ref{simplify:no-neg-L}, the imbalanced instances are most constraining
\[
    - s_d + 1/2 \cdot r_d \leq 0 \text{ where } d \in \{2,3\} \text{ and } -s_3' + 1/2 \cdot r_3 \leq 0
\]
but this is no more constraining than line \ref{simplify:no-neg-L}.

Line \ref{simplify:bal-drag} drags to $R$ the degree-2 separator vertex $s$ and a 2-path, ending in a vertex $l$ which is either in $S$ or has degree 3, which itself is dragged into $S$. This most constraining in the balanced case
\begin{equation*}
    -s_2 + s_3' + \frac{1}{2} \cdot (r_2 - r_3) \leq 0;
    \label{eq:pathdraggingright}
\end{equation*}
In line \ref{simplify:imbal-drag} the most constraining case is
\begin{equation*}
    -s_2 + s_3' - r_3 \leq 0 \enspace.
\end{equation*}

The operations in line \ref{simplify:no-l-skeleton} drag neighbors and associated 2-paths from $L$ into $R$, also removing $s \in S$. Since $r_2 = 0$ we can simplify the most constraining case, which is imbalanced, to: $-s_3 + 2 r_3\leq 0$. Line \ref{simplify:no-r-skeleton} is most constraining in the balanced case, which induces the constraint $-s_3 \leq 0$.

\begin{claim}
\label{claim:always-2-r3-in-2nd-neg}
    After \texttt{simplify} has been applied to a graph $G$ and it's separation $(L,S,R)$, for $s \in S$ there exists $r \in N_{\Gamma}(s) \cap R$ such that $N_{\Gamma}(r) \cap R \neq \emptyset$, and also there exists $l \in N_{\Gamma}(s) \cap L$ such that $N_{\Gamma}(l) \cap L \neq \emptyset$
\end{claim}
\begin{proof}
    If there is a vertex $s$ that does not satisfy the claim, then line \ref{simplify:no-l-skeleton} or \ref{simplify:no-r-skeleton} would trigger and remove $s$ from $S$.
\end{proof}

\subparagraph{Constraints from \texttt{spider}}
The first two conditions of lines \ref{spider:pull-r} and \ref{spider:pull-l} in \texttt{spider} aim to drag into the separator a (2,2,3) or (2,3,3) vertex in order to branch more efficiently on. In the worst case there is no change in measure since $s$ is replaced by $r$ in the separator. Since the separation $(L,S,R)$ is balanced, moving $P_r$ and $r$ or $P_l$ and $l$ also does not change the measure as $L$ and $R$ contribute equally to $\mu_{8/3}$.

In line \ref{spider:refineLSR-spider-drag}, $s$ is a center spider vertex with attributed weight $s_3'$. We branch on $l \in L$, which is a skeleton neighbor of $s$. The for loop drags vertices which are skeleton neighbors of $l$ with no change in measure so that when $l$ is branched on, it obtains a decrease in measure of at least $\frac{3}{2}r_3$ by it's neighbors. However, we choose $l$ to branch on because on both subproblems, branching on $l$ causes the removal of $s$ from the separator as it no longer has neighbors in $L$. This results in the branching constraint:
\[ \left(s_3' + \frac{1}{2}(r_3 + 2\Delta r_3), s_3' + \frac{1}{2}(r_3 + 2\Delta r_3) \right) \enspace.\]

Line \ref{spider:branch-s} finds a valid left or right spider vertex and branches on it, resulting in the constraints
\[ \left(s_3' + \frac{3}{2}\Delta r_3, s_3' + \frac{3}{2}\Delta r_3 \right) \enspace.\]

\subparagraph{Constraints from \texttt{\#3IS} - Computing Separator.}
Much like in \cite{GaspersS15}, computing a new separator in line \ref{3is:new-sep} of \texttt{\#3IS} imposes the constraint
\begin{equation*}
    \begin{aligned}
    &s_3' / 6 + 5/12 \cdot r_3 < r_3 \text{, or} 
    &s_3' < 7/2 \cdot r_3.
    \end{aligned}
    \label{eq:newseparator} 
\end{equation*}  
In line \ref{3is:simplifycall} the algorithm simplifies the graph $G$ and it's separation $(L,S,R)$ through a call to \texttt{simplify}, which itself imposes new constraints. 

The reduction rule in line \ref{3is:imbaldrag-tobranch} is the same as the constraints for line \ref{simplify:imbal-drag} in \texttt{simplify}. We now deal with branching on lazy-2 separators and regular branching, in both imbalanced and balanced cases, separately. As decreasing a degree-3 vertex to a degree-2 vertex may result in the introduction of a spider vertex $s_3'$ from $s_3$, let $\delta = s_3' - s_3$ be the increase in measure from a spider vertex creation, offset by either a $\Delta s_3$ or $\frac{1}{2}\Delta r_3$ decrease in measure.

\subparagraph{Constraints from \texttt{\#3IS} - Balanced Lazy 2-Separator Branching}
Suppose the instance is balanced and \texttt{\#3IS} selects a vertex $s \in S$ but $s$ has a lazy 2-separator $\{y,z\}$ which line \ref{3is:sep2} of \texttt{\#3IS} branches on instead of $s$.
As the degree-3 vertices $y$, $z$ and $s$ are all removed in the branches of this problem, as well as the fact that due to Claim \ref{claim:always-2-r3-in-2nd-neg} for $L$ and $R$ there will be another degree-3 vertex that will be removed, we obtain the branching vector
\[ \left(s_3 + \frac{1}{2}(2 r_3 + 2\Delta r_3) - 2 \delta, s_3 + \frac{1}{2}(2 r_3 + 2\Delta r_3) - 2 \delta \right) \enspace.\]
The worst case contains measure increases of $2 \delta$ since the two decreases of $\frac{1}{2}\Delta r_3$ could create a spider vertex, and there are at least 2 of them. We could have more $\delta$ decreases, but this only occurs when we have a $\frac{1}{2} \Delta r_3$ decrease, or $\Delta s_3$ decrease in the worst case. But since $\delta \leq \Delta s_3 \leq \frac{1}{2} \Delta r_3$ the tightest constraint occurs at the smallest number of $\delta$ possible.

\subparagraph{Constraints from \texttt{\#3IS} - Imbalanced Lazy 2-Separator Branching}
Once again, we have vertices $s \in S$ and a lazy 2-separator $\{y,z\}$, but the instance is imbalanced. First assume either 1 or more of $\{y,z\}$ is in $R$. In this case, we disconnect $s$, a $y$ or $z$, as well as some other vertex $r \in R$ due to Claim \ref{claim:always-2-r3-in-2nd-neg}. At worst this results in the branching vector $(s_3 + 2 r_3, s_3 + 2 r_3)$

In the case where $\{y,z\} \in L$ also divert to Claim 1 which guarantees that there is a skeleton neighbor $r \in N_{\Gamma}(s) \cap R$, which itself has a neighbor $r' \in N_{\Gamma}(r) \cap R$. These two combined with $s$ are removed in both branches, otherwise $s$ cannot be removed and $\{y,z\}$ is not a lazy-2 separator. This also results in the branching vector $(s_3 + 2 r_3, s_3 + 2 r_3)$


\subparagraph{Constraints from \texttt{\#3IS} - Balanced Branching: neighbor in separator} 
Consider the balanced branching case where we branch on $s \in S$ and $s$ has a neighbor $s' \in S$. 
Let $u \in R$ and $v \in L$ denote the two other neighbors. In the worst case, $u$ and $v$ are both degree-2 vertices, meaning in both branches we only reduce a vertex of weight $r_3$ to $r_2$, but never delete one. Since $s'$ reduces in degree in the first branch and is removed in the second branch, we get the following branching vector
\[ \left( s_3 + \Delta s_3 + \frac{1}{2}(2 \Delta r_3) - 3 \delta, 2s_3 + \frac{1}{2}(2 \Delta r_3) - 2 \delta\right)  \enspace.\]

\subparagraph{Constraints from \texttt{\#3IS} - Balanced Branching: no neighbor in separator.} Next consider the balanced branching case where the algorithm branches on a non-spider vertex $s \in S$ with no neighbors in the separator $S$. Let $u, u' \in R$ and $v \in L$ denote its neighbors. Since $s$ is a non-spider vertex then $s$ is either a (2,2,3) or (2,3,3) vertex.

We first consider $s$ as a (2,2,3) vertex. In the worst case, the single degree-3 vertex of weight $\frac{r_3}{2}$ would be in $R$ or $L$ since a weight of $s_3 > r_3$, and in practice it is much larger. Of the two remaining neighbors, they are the start of a 2-path to another degree-3 vertex. Now both of these cannot be in $S$ so we will have a decrease of at least $\frac{\Delta r_3}{2}$, leaving a decrease of $\Delta s_3$ for the last neighbor. 

In the second case, we also get a decrease of $\Delta s_3 + \frac{\Delta r_3}{2}$ from the degree-3 neighbor of $s$. This is due to Claim \ref{claim:always-2-r3-in-2nd-neg} forcing at least 1 of the neighbors to be in $R$. This results in a branching vector of
\[\left(s_3 + \Delta s_3 + \frac{1}{2}( 2\Delta r_3) - 3 \delta, s_3 + 2 \Delta s_3 + \frac{1}{2}(r_3 + 2 \Delta r_3 ) - 4 \delta \right)  \enspace.\]

Now if $s$ is a (2,3,3) vertex, we get 2 degree 3 neighbors of $s$. In the worst case, the degree 2 neighbor of $s$ is the start of a 2-path to another vertex in $S$.
\[ \left( s_3 + \Delta s_3 + \frac{1}{2}(2 \Delta r_3) - 3 \delta, s_3 + 3 \Delta s_3 + \frac{1}{2}(2 r_3 + 2 \Delta r_3) - 5 \delta \right)  \enspace.\]

\subparagraph{Constraints from \texttt{\#3IS} - Imbalanced Branching: neighbor in separator.} 
In the imbalanced instances of $G$ the measure $\mu_{8/3}$ simplifies to $\mu_{\nicefrac{8}{3}} = \mu_s(S) + \mu_r(R) + \mu_o(L,S,R)$.
Suppose we choose $s \in S$ to branch on and $s$ has a neighbor $s' \in S$. By Claim \ref{claim:always-2-r3-in-2nd-neg}, $s$ has a skeleton neighbor $r \in N_{\Gamma}(s) \cap R$. Now in the worst case, $r$ is only a skeleton neighbor, and the actual neighbor $r' \in N_{G}(s) \cap R$ is of degree 2. By considering the removal, or reduction of degree, of $s$, $s'$ and $r'$ then we get the following worst case constraint
\[(s_3 + \Delta s_3 + r_3 - 3 \delta, 2 s_3 + r_3 + 5 \delta) \enspace.\]
The first branch has a $3 \delta$ term since we get at most 1 decrease for each neighbor. The $5 \delta$ term comes from the fact that the left neighbor $l \in N_{G}(s) \cap L$ does not contribute any weight to $\mu_{8/3}$ meaning it could be degree 3. Now $s'$ is also of degree 3, so in the second case where we remove $s'$ and $l$, these two could create 4 spider vertices. The last possible increase comes from $r$ being reduced to a degree-2 vertex.


\subparagraph{Constraints from \texttt{\#3IS} - Imbalanced Branching: no neighbors in separator.} There are two branching rules to consider in this case. First first branching occurs in line \ref{3is:bad_imbal_case} where instead of branching on $s \in S$ we branch on one of its skeleton neighbors in $R$. The other case occurs when we branch on $s$ as normal in line \ref{3is:branch}.

In line \ref{3is:bad_imbal_case}, we are given the case where $s$ has 1 skeleton neighbor in $R$. This means that we don't get a beneficial branching by branching on $s$. However, in a similar method to line \ref{spider:branch-l} of \texttt{spider}, if we branch on $r \in N_{\Gamma}(s) \cap R$ such that $N_{\Gamma}(r) \cap R \neq \emptyset$, then in both branches, we are able to remove $s$ entirely from the separator due to the simplification rules in \texttt{simplify}. We get the following worst case constraint
\[(r_3 + s_3 + \Delta r_3 + \Delta s_3 - 3 \delta, r_3 + s_3 + \Delta r_3 + \Delta s_3 - 3 \delta)  \enspace.\]

Otherwise, we progress to line \ref{3is:branch}, which guarantees that we have 2 skeleton neighbors of $s$ in $R$. This results in the following constraint
\[(s_3 + 2 \Delta r_3 - 3 \delta, s_3 + 2 \Delta r_3 - 4 \delta)  \enspace.\]

} 
\longversion{\proofLemmaEightThirdRunningTime}



\newcommand{\proofWeightsEightThirdRunningTime}{
\subparagraph{Weights and Results.}

The combination of all constraints obtained in this way, minimizing the measure results in the measure of $\mu_{\nicefrac{8}{3}} = 0.13262 \cdot n$, and that the running time is $O(2^{\mu_{\nicefrac{8}{3}}}) \subseteq O(2^{0.13262 n})$ results in an upper bound of $O(\rteightthirds^n)$. The specific weights are summarized below.
\shortversion{
\[
    r_0 = r_1 = r_2 = 0,\quad r_3 = 0.2 + o(n),\quad s_0 = s_1 = 0,\quad s_2 = 0.6,\quad  s_3 = 0.6838,\quad s_3' = 0.7
\]
}
\longversion{
\[\begin{aligned}
    r_0 = 0 &\quad r_1 = 0 &r_2 = 0 \quad\quad\text{ } &\quad r_3 = 0.2 + o(n) &\\
    s_0 = 0 &\quad s_1 = 0 &\quad s_2 = 0.6352 &\quad s_3 = 0.6784 \quad & s_3' = 0.7\\
\end{aligned}\]
}
}
\proofWeightsEightThirdRunningTime 

\begin{lemma}[] \label{lem:degree3runtime}
    Algorithm \texttt{\#IS} applied to a graph $G$ with $d(G) \leq 3$ has running time $O(\rtcubic^n)$ and uses polynomial space.
\end{lemma}
The algorithm \texttt{\#IS} uses subroutine \texttt{\#3IS}, which we analyze the measure and the weights for. 
We equate the Separate, Measure and Conquer weights with weights of the measure $\mu_3$, based on the compound analysis from Wahlstr\"{o}m \cite{Wahlstrom08}. 
As Wahlstr\"{o}m's analysis only contains weights $w_3'$ and $w_2'$, for vertices of degree 3 and degree 2 respectively, the measure is
\[ 
  \mu_3(G) = \left((d - 2) w_3' + (3 - d) w_2' \right)n + \mu_o(L,S,R)
\]
where $d = d(G)$ is the average degree of a cubic graph, and $\mu_o(L,S,R)$ is the sub-linear term left over from the average degree $\nicefrac{8}{3}$ analysis. 

In the case of a graph $G$ with no (3,3,3) vertex, in order for Lemma \ref{lem:combinemeasureanalysis} to apply, the values of $w_1$ and $w_2$ must satisfy inequalities
\shortversion{
  $
  \frac{r_2}{2} \leq w_2,
  $
}
\longversion{
\[
    \frac{1}{2}r_2 \leq w_2,
\]
}
\shortversion{
  $
  \frac{s_3'}{11} + \frac{5 r_3}{22} + \frac{5 r_2}{22} \leq \frac{6 w_3}{11} + \frac{5 w_2}{11},
  $
}
\longversion{
\[
  \frac{s_3'}{11} + \frac{5 r_3}{22} + \frac{5 r_2}{22} \leq \frac{6 w_3}{11} + \frac{5 w_2}{11},
\]
}
\shortversion{
  $
  \frac{s_3}{9} + \frac{5 r_3}{18} + \frac{r_2}{3} \leq \frac{2 w_3}{3} + \frac{ w_2}{3},
  $
}
\longversion{
\[
  \frac{s_3}{9} + \frac{5 r_3}{18} + \frac{r_2}{3} \leq \frac{2 w_3}{3} + \frac{ w_2}{3},
\]
}
induced when $d = 2, \frac{28}{11}$, and $\frac{8}{3}$ for $\mu_{8/3}$ respectively. This results in the weights $w_3 = 0.1973$ and $w_2 = 0.0033$ when $G$ has no (3,3,3) vertex.

We also let $w_3' \geq 0$ and $w_2' \geq 0$ be the weights associated with vertices of degree 3 and degree 2 respectively, for a subcubic graph $G$. Using the analysis by compound measures with $\mu_3(G) = \sum_{i \in \{2,3\}} w_i' \cdot n_i$, the following constraint
\shortversion{$\mu_{\nicefrac{8}{3}}(G) = \mu_{3}(G)$ when $d(G) = \nicefrac{8}{3}$}
\longversion{
\[
    \mu_{\nicefrac{8}{3}}(G) \leq \mu_{3}(G), \quad \text{ when $d(G) = \nicefrac{8}{3}$ }
\]
}
is required for a valid compound measure. This can be rewritten as
\shortversion{$2 w_3 + w_2 = 2 w_3' + w_2'$.}
\longversion{
\begin{equation*}
    2 w_3 + w_2 \leq 2 w_3' + w_2'.
\end{equation*}
}
Branching on a (3,3,3) vertex, the only type of degree-3 vertex that will be branching in \texttt{\#IS}, gives a branching vector of 
\shortversion{$(4w_3' - 3w_2', 8w_3' - 4w_2')$.} 
\longversion{
\begin{equation*}
    (4w_3' - 3w_2', 8w_3' - 4w_2').
\end{equation*}
}
Setting the weights $w_3' = 0.1876$ and $w_2' = 0.0228$ satisfies the system of constraints described above and by using the measure $\mu_3(G)$, results in a running time of $O^*(\rtcubic^n)$.



\newcommand{\degreeFourAnalysis}{
\subsection{Degree-4 Analysis}

\label{subsec:Degree4}
For a graph with maximum degree 4, analysis is done with a measure of 
\[\mu_4 = \sum_{i \leq 4} w_i \cdot n_i + \mathds{1} \left(
\parbox{20em}{$G$ has only degree-4 and degree-2 vertices and no degree-4 vertex has a degree-4 neighbor}
\right)\psi + \mu_o(L,S,R)\]
where $w_i$ are weights attributed to vertices of degree $i$, $n_i$ are the number of vertices with degree $i$ and $\psi$ is a potential.
We can ignore lower weights since due to simplification rules we have for vertices of with degree 0 or 1.

\subparagraph{Potentials in Degree-4 Analysis. } Potentials are used for branching on a degree-4 vertex $v$ with only degree-2 neighbors. In case (a), we have that all 2-paths starting from $v$, have endpoints of degree 4. Case (b) has at least one 2-path from $v$ that ends up in a degree-3 vertex.

\begin{figure}[ht!]
    \centering
    \begin{tabular}{|c|c|c|}
    \hline
    Degrees of Neighbors & Highest Average Degree & Branching \\
    \hline
    (2,2,2,2) (a) &             3     & $\tau(5w_4 - 4w_3 +4w_2 - \psi, 5w_4 - 4w_3 +4w_2 - \psi) $ \\
    (2,2,2,2) (b) &             3     & $\tau(4w_4 - 2w_3 +3w_2 + \psi, 4w_4 - 2w_3 +3w_2 + \psi) $ \\
    (2,2,2,3) &                 3     & $\tau(4w_4 - 2w_3 +2w_1, 4w_4 - 2w_3 +3w_2) $ \\
    (2,2,2,4) &                 3     & $\tau(5w_4 - 4w_3 +3w_2, 6w_4 - 4w_3 +3w_2) $ \\
    (2,2,3,3) &                 3     & $\tau(3w_4, 5w_4 - 2w_3 +2w_2)$ \\
    (2,2,3,4) &                 3     & $\tau(4w_4 - 2w_3 +w_2, 5w_4 - 2w_3 +2w_2) $ \\
    (2,2,4,4) &                 3     & $\tau(5w_4 - 4w_3 +2w_2, 7w_4 - 4w_3 +2w_2) $ \\
    (2,3,3,3) & 16/5          = 3.2   & $\tau(2w_4 +2w_3 - 2w_2,4w_4 +w_2) $ \\
    (2,3,3,4) & 42/13 $\approx$ 3.23  & $\tau(3w_4 - w_2, 6w_4 - 2w_3 +w_2) $ \\
    (2,3,4,4) & 36/11 $\approx$ 3.27  & $\tau(4w_4 - 2w_3, 6w_4 - 2w_3 +w_2) $ \\
    (2,4,4,4) & 10/3  $\approx$ 3.33  & $\tau(5w_4 - 4w_3 +w_2, 8w_4 - 4w_3 +w_2) $ \\
    (3,3,3,3) & 24/7  $\approx$ 3.43  & $\tau(w_4 +4w_3 - 4w_2, 5w_4)$ \\
    (3,3,3,4) &  7/2          = 3.5   & $\tau(2w_4 + 2w_3 -  3w_2, 5w_4)$ \\
    (3,3,4,4) & 18/5          = 3.6   & $\tau(3w_4 - 2w_2, 7w_4 - 2w_3)$ \\
    (3,4,4,4) & 15/4          = 3.75  & $\tau(4w_4 -  2w_3 -  w_2, 7w_4 -  2w_3)$ \\
    (4,4,4,4) & 4 & $\tau(5w_4 - 4w_3, 9w_4 - 4w_3)$ \\
    \hline
    \end{tabular}
    \caption{Possible cases when branching on a degree-4 vertex}
    \label{fig:deg4table}
\end{figure}
}
\longversion{\degreeFourAnalysis}

\begin{restatable}{lemma}{lemWahlDegFour}\label{lem:wahlDegFour}
    For a graph $G$ with maximum degree 4, \#IS can be solved in time $O^*(\rtquad^n)$.
\end{restatable}

\newcommand{\proofWahlDegFour}{
\begin{proof} 
    The degree-4 analysis uses pivot points 3, 3.2, 3.5, 3.75 and 4, shown as different rows of Figure \ref{fig:deg4measure}. 
    Pivot points generate multiple compound measures with weights and constraints for each. 
    By including constraints generated from the table of branching factors in Figure \ref{fig:deg4table}, we gain satisfying weights for $\mu_4$, shown in Figure \ref{fig:deg4measure}.
    This results in a running time upper bound of $O(2^{\mu_4 n}) \subseteq O(2^{0.2713 n}) \subseteq O(\rtquad^n) $ in the worst case for degree-4 graphs.
\end{proof}

\begin{figure}[th!]
    \centering
    \begin{tabular}{|c|c|c|c|c|}
        \hline
        Average Degree & $w_2$ & $w_3$ & $w_4$ & Time \\
        \hline
        2-3      & 0.0227913 & 0.1875202 & 0.3295266 & $O(1.13880^n)$ \\
		3-3.2    & 0.0659881 & 0.1875202 & 0.2863298 & $O(1.15451^n)$ \\
		3.2-3.5  & 0.0795475 & 0.1897802 & 0.2772902 & $O(1.17571^n)$ \\
		3.5-3.75 & 0.0911988 & 0.1936639 & 0.2734064 & $O(1.19207^n)$ \\
		3.74-4   & 0.1057321 & 0.1998925 & 0.2713302 & $O(\rtquad^n)$ \\
        \hline
    \end{tabular}
    \caption{Component measures $\sum_i w_i \cdot n_i$ for maximum degree 4}
    \label{fig:deg4measure}
\end{figure}
}
\longversion{\proofWahlDegFour}

\newcommand{\degreeFiveAnalysis}{
\subsection{Degree-5+ Analysis}

The following two theorems show for degree-5+ graphs the generalized procedure for constructing branching vectors for $v$ and all its possible combinations of degrees of neighbors.

\begin{lemma}

Suppose a graph $G$ is 3-connected. Let $v \in V(G)$ be a vertex to be branched on in $\texttt{\#IS}$ with $d(v) = \{5, 6\}$. Let $out(v)$ be  the number of outgoing edges of type $(u,u')$ such that $u \in N(v)$ and $u' \notin N(v) \cup \{v\}$. Then
\begin{align*}
    \label{eq:outformula}
    out(v) = 
    \begin{cases}
        3 & \text{If }   d(v) = 5 \text{ and } \sum_{u \in N(v)} d(u) = 0 \text{ mod } 2 \text{ or } \\
          & \text{\quad} d(v) = 6 \text{ and } \sum_{u \in N(v)} d(u) = 1 \text{ mod } 2 \\
        4 & \text{If }   d(v) = 5 \text{ and } \sum_{u \in N(v)} d(u) = 1 \text{ mod } 2 \text{ or } \\
          & \text{\quad} d(v) = 6 \text{ and } \sum_{u \in N(v)} d(u) = 0 \text{ mod } 2 \\
        5 & \text{ If neighbors of $v$ have degree (2, 2, 2, 2, 2) or (2, 2, 2, 2, 2, 3)} \\
        6 & \text{ If neighbors of $v$ have degree (2, 2, 2, 2, 2, 2)}. \\
    \end{cases}
\end{align*}
\end{lemma}

\begin{proof}
    Let $out(v)$ represent the minimum number of outgoing edges $xy$ from $N(v)$ with $x \in N(v)$ and $y \notin N(v)$. 
    We suppose we have a 3-connected graph with all simplification rules applied. 
    This means that there are multiplier reduction does not apply, and there are no lazy 2-separators. 
    If $out(v) = 0$ then we have an instance of constant size, which can be solved in constant time. 
    If $out(v) = 1$ we can apply the multiplier reduction, which is a contradiction.
    Similarly, if $out(v) = 2$ we have a lazy 2-separator which is also a contradiction.
    Hence $out(v) \geq 3$. 

    Suppose $d(v) = 5$ and $v$ has neighbors with degrees (2, 2, 2, 2, 2). 
    Any edge adjacent to two neighbors of $v$ means $G$ can be reduced by multiplier reduction by branching on $v$, so $out(v) = 5$ Similarly, if $d(v) = 6$ and $v$ has neighbors (2, 2, 2, 2, 2, 2), then $out(v) = 6$.

    Supppose $d(v) = 6$ and $v$ has neighbors (2, 2, 2, 2, 2, 3).  
    There are 7 edges adjacent to $N(v)$ but not $v$. 
    Suppose $u \in N(v)$ and $d(u) = 3$. 
    If $out(v) < 5$ then at least 3 of these 7 edges must connect two vertices in $N(v)$, but at most two of them are adjacent to $u$. 
    Thus there exists one edge $\{a,b\}$ with $d(a) = d(b) = 2$. But this means multiplier reduction can be applied, hence $out(v) = 5$.

    Suppose $d(v) = 5$ and $\sum_{u \in N(v)} d(u) = 1 \mod 2$. We showed there are at least 3 outgoing edges from $N(v)$. 
    There are also 5 edges adjacent to $N(v)$ and $v$ which gives a total of at least 8 edges that are adjacent to $N(v)$. Since having adjacent neighbors does not change the fact that $\sum_{u \in N(v)} d(u)$ is odd, $out(v)$ must be even.

    If $\sum_{u \in N(v)} d(u)$ is odd, then 
    since any edge adjacent to two neighbors of $v$ contributes a value of 2 to the sum, then $\sum_{u \in N(v)} d(u) = 1$ mod 2 implies $out(v)$ = 4. A similar parity argument is used for $d(v) = 6$, except with the parity swapped around.
    
\end{proof}

\begin{lemma}
Let $deg_2(v)$ denote the number of degree-2 vertices in $N(v)$. Then $v$ has a branching vector of
\begin{equation}
    \label{thm:wahldeg5plusformula}
    \left (w_{d(v)} + \sum_{u \in N(v)} w_{d(u)} + out(v) \cdot \Delta w_{d(v)}, \quad
                w_{d(v)} + \sum_{u \in N(v)} \Delta w_{d(u)} + deg_2(v) \cdot \Delta w_{d(v)} \right).
\end{equation}  
\end{lemma}

\begin{proof}
    The left hand side of the branching factor considers removing a vertex $v$ and it's neighbors. The right hand side considers removing a just vertex $v$. The reduction in measure on the graph $G$ follows from reduction rules, the measure $\mu = \sum_{v \in F}w_{d(v)} + \mu_o(L,S,R)$ and the definition of $out(v)$ and $deg_2(v)$.
\end{proof}

}
\longversion{\degreeFiveAnalysis}

\begin{restatable}{theorem}{lemWahlDegFive}\label{lem:wahlDegFive}
    \label{lem:wahldegall}
    \#IS can be solved in time $O^*(\rtis^n)$ and polynomial space.
\end{restatable}

\newcommand{\proofWahlDegFive}{
\begin{proof}
If  $d(G) \geq 7$ we can perform a quick analysis in terms of $n$, and the branching number is at worst $\tau (1, 8) < 1.2321$. So we only need to compute $\mu_6(G)$ with compound measures using Equation \ref{thm:wahldeg5plusformula}, with $d(G) \leq 6$ in order to find the worst case running time for \#IS.
\begin{figure}[ht]
    \centering
    \begin{tabular}{|c|c|c|c|c|c|c|}
    \hline
       Average Degree & $w_2$   &    $w_3$   &    $w_4$   &    $w_5$   &    $w_6$   &     Time      \\
    \hline
    4-6 & 0.1146078 & 0.2017931 & 0.2713406 & 0.2977566 & 0.3051140 & $O(\rtis^n)$ \\
    \hline
    \end{tabular}
    \caption{Weights and running time for $\mu_6(G)$}
\end{figure}

\end{proof}

}
\longversion{\proofWahlDegFive}

If we plug in a simple pathwidth-based subroutine \cite{FominH06} for graphs of maximum degree $3$, we obtain the following exponential-space result.

\begin{theorem}
	\#IS can be solved in time $O^*(\esrtis^n)$.
\end{theorem}

\subparagraph*{Acknowledgements}

We thank Magnus Wahlstr{\"o}m for clarifying an issue of the case analysis in \cite{Wahlstrom08} and an anonymous reviewer for useful comments on an earlier version of the paper.
Serge Gaspers is the recipient of an Australian Research Council (ARC) Future Fellowship (FT140100048) and acknowledges support under the ARC's Discovery Projects funding scheme (DP150101134).


\shortversion{
    \appendix
    \newpage
    \section{Additional details and proofs}

    \restateLemma{\lemSep}{lem:sep}
    \proofLemSep

    \restateLemma{\lemUpperBound}{lem:upperBound}
    \proofLemUpperBound

    \restateLemma{\lemEightThirdRunningTime}{lem:eightThirdRunningTime}
    \proofLemEightThirdRunningTime
    \proofWeightsEightThirdRunningTime




    \degreeFourAnalysis

    \restateLemma{\lemWahlDegFour}{lem:wahlDegFour}
    \proofWahlDegFour

    \degreeFiveAnalysis

    \restateLemma{\lemWahlDegFive}{lem:wahlDegFive}
    \proofWahlDegFive
} 
  


\bibliography{bibliography}

\begin{thebibliography}{10}

\bibitem{AngelsmarkT05}
Ola Angelsmark and Johan Thapper.
\newblock Partitioning based algorithms for some colouring problems.
\newblock In {\em Recent Advances in Constraints, Joint ERCIM/CoLogNET
  International Workshop on Constraint Solving and Constraint Logic Programming
  ({CSCLP} 2005)}, volume 3978 of {\em Lecture Notes in Computer Science},
  pages 44--58. Springer, 2005.

\bibitem{BjorklundH06}
Andreas Bj{\"{o}}rklund and Thore Husfeldt.
\newblock Inclusion--exclusion algorithms for counting set partitions.
\newblock In {\em Proceedings of the 47th Annual {IEEE} Symposium on
  Foundations of Computer Science ({FOCS} 2006)}, pages 575--582. {IEEE}
  Computer Society, 2006.

\bibitem{BjorklundH08}
Andreas Bj{\"{o}}rklund and Thore Husfeldt.
\newblock Exact algorithms for exact satisfiability and number of perfect
  matchings.
\newblock {\em Algorithmica}, 52(2):226--249, 2008.

\bibitem{BjorklundHK09}
Andreas Bj{\"o}rklund, Thore Husfeldt, and Mikko Koivisto.
\newblock Set partitioning via inclusion-exclusion.
\newblock {\em SIAM Journal on Computing}, 39(2):546--563, 2009.

\bibitem{BodlaenderK96}
Hans~L. Bodlaender and Ton Kloks.
\newblock Efficient and constructive algorithms for the pathwidth and treewidth
  of graphs.
\newblock {\em Journal of Algorithms}, 21(2):358--402, 1996.

\bibitem{BodlaenderK06}
Hans~L. Bodlaender and Dieter Kratsch.
\newblock An exact algorithm for graph coloring with polynomial memory.
\newblock Technical Report UU-CS-2006-015, Department of Information and
  Computing Sciences, Utrecht University, 2006.

\bibitem{Byskov04}
Jesper~Makholm Byskov.
\newblock Enumerating maximal independent sets with applications to graph
  colouring.
\newblock {\em Operations Research Letters}, 32(6):547--556, 2004.

\bibitem{Christofides71}
Nicos Christofides.
\newblock An algorithm for the chromatic number of a graph.
\newblock {\em The Computer Journal}, 14(1):38--39, 1971.

\bibitem{DahllofJW05}
Vilhelm Dahll{\"{o}}f, Peter Jonsson, and Magnus Wahlstr{\"{o}}m.
\newblock Counting models for 2{SAT} and 3{SAT} formulae.
\newblock {\em Theoretical Computer Science}, 332(1-3):265--291, 2005.

\bibitem{Eppstein01}
David Eppstein.
\newblock Small maximal independent sets and faster exact graph coloring.
\newblock In {\em Proceedings of the 7th International Workshop on Algorithms
  and Data Structures (WADS 2001)}, volume 2125 of {\em Lecture Notes in
  Computer Science}, pages 462--470. Springer, 2001.

\bibitem{Eppstein03}
David Eppstein.
\newblock Small maximal independent sets and faster exact graph coloring.
\newblock {\em Journal of Graph Algorithms and Applications}, 7(2):131--140,
  2003.

\bibitem{FederM02}
Tom{\'{a}}s Feder and Rajeev Motwani.
\newblock Worst-case time bounds for coloring and satisfiability problems.
\newblock {\em Journal of Algorithms}, 45(2):192--201, 2002.

\bibitem{fomin2009two}
Fedor~V Fomin, Serge Gaspers, Saket Saurabh, and Alexey~A Stepanov.
\newblock On two techniques of combining branching and treewidth.
\newblock {\em Algorithmica}, 54(2):181--207, 2009.

\bibitem{FominGK09}
Fedor~V. Fomin, Fabrizio Grandoni, and Dieter Kratsch.
\newblock A measure {\&} conquer approach for the analysis of exact algorithms.
\newblock {\em Journal of the {ACM}}, 56(5), 2009.

\bibitem{FominH06}
Fedor~V Fomin and Kjartan H{\o}ie.
\newblock Pathwidth of cubic graphs and exact algorithms.
\newblock {\em Information Processing Letters}, 97(5):191--196, 2006.

\bibitem{FominK10}
Fedor~V. Fomin and Dieter Kratsch.
\newblock {\em Exact exponential algorithms}.
\newblock Springer Science \& Business Media, 2010.

\bibitem{FurerK07}
Martin F{\"{u}}rer and Shiva~Prasad Kasiviswanathan.
\newblock Algorithms for counting 2-{S}at solutions and colorings with
  applications.
\newblock In {\em proceedings of the 3rd International Conference on
  Algorithmic Aspects in Information and Management ({AAIM} 2007)}, volume 4508
  of {\em Lecture Notes in Computer Science}, pages 47--57. Springer, 2007.

\bibitem{Gaspers10}
Serge Gaspers.
\newblock {\em Exponential Time Algorithms}.
\newblock VDM Verlag, 2010.

\bibitem{GaspersS12}
Serge Gaspers and Gregory~B. Sorkin.
\newblock A universally fastest algorithm for {M}ax 2-{S}at, {M}ax 2-{CSP}, and
  everything in between.
\newblock {\em Journal of Computer and System Sciences}, 78(1):305--335, 2012.

\bibitem{GaspersS15}
Serge Gaspers and Gregory~B. Sorkin.
\newblock Separate, measure and conquer: Faster polynomial-space algorithms for
  {M}ax 2-{CSP} and counting dominating sets.
\newblock In {\em Proceedings of the 42nd International Colloquium on Automata,
  Languages, and Programming ({ICALP} 2015)}, volume 9134 of {\em Lecture Notes
  in Computer Science}, pages 567--579. Springer, 2015.

\bibitem{greenhill2000complexity}
Catherine Greenhill.
\newblock The complexity of counting colourings and independent sets in sparse
  graphs and hypergraphs.
\newblock {\em Computational Complexity}, 9(1):52--72, 2000.

\bibitem{Iwata11}
Yoichi Iwata.
\newblock A faster algorithm for dominating set analyzed by the potential
  method.
\newblock In {\em Proceedings of the 6th International Symposium on
  Parameterized and Exact Computation ({IPEC} 2011)}, volume 7112 of {\em
  Lecture Notes in Computer Science}, pages 41--54. Springer, 2011.

\bibitem{Junosza-SzaniawskiT15}
Konstanty Junosza-Szaniawski and Michal Tuczynski.
\newblock Counting independent sets via divide measure and conquer method.
\newblock Technical Report abs/1503.08323, arXiv CoRR, 2015.

\bibitem{Koivisto06}
Mikko Koivisto.
\newblock An ${O}^*(2^n)$ algorithm for graph coloring and other partitioning
  problems via inclusion--exclusion.
\newblock In {\em Proceedings of the 47th Annual {IEEE} Symposium on
  Foundations of Computer Science ({FOCS} 2006)}, pages 583--590. {IEEE}
  Computer Society, 2006.

\bibitem{Lawler76}
Eugene~L. Lawler.
\newblock A note on the complexity of the chromatic number problem.
\newblock {\em Information Processing Letters}, 5(3):66--67, 1976.

\bibitem{roth1996hardness}
Dan Roth.
\newblock On the hardness of approximate reasoning.
\newblock {\em Artificial Intelligence}, 82(1):273--302, 1996.

\bibitem{scott2009polynomial}
Alexander~D Scott and Gregory~B Sorkin.
\newblock Polynomial constraint satisfaction problems, graph bisection, and the
  ising partition function.
\newblock {\em ACM Transactions on Algorithms (TALG)}, 5(4):45, 2009.

\bibitem{Wahlstrom04}
Magnus Wahlstr{\"{o}}m.
\newblock Exact algorithms for finding minimum transversals in rank-3
  hypergraphs.
\newblock {\em Journal of Algorithms}, 51(2):107--121, 2004.

\bibitem{Wahlstrom08}
Magnus Wahlstr{\"{o}}m.
\newblock A tighter bound for counting max-weight solutions to 2{SAT}
  instances.
\newblock In {\em Proceedings of the 3rd International Workshop on
  Parameterized and Exact Computation ({IWPEC} 2008)}, volume 5018 of {\em
  Lecture Notes in Computer Science}, pages 202--213. Springer, 2008.

\bibitem{Woeginger08}
Gerhard~J. Woeginger.
\newblock Open problems around exact algorithms.
\newblock {\em Discrete Applied Mathematics}, 156(3):397--405, 2008.

\end{thebibliography}

\end{document}